\documentclass[12pt, draftclsnofoot, journal, onecolumn,romanappendices]{IEEEtran}

\usepackage[pdftex]{graphicx}
\usepackage{bm,cite}
\usepackage{amsmath} 
\usepackage{amsthm}
\usepackage{amsbsy}
\usepackage{epsfig} 
\usepackage{latexsym} 
\usepackage{amssymb}
\usepackage{wasysym}
\usepackage{placeins}
\usepackage{color}
\usepackage{rotating}
\usepackage[usenames,dvipsnames]{xcolor}
\usepackage{dsfont}
\usepackage{caption}
\usepackage{subcaption}
\usepackage{algorithm}
\usepackage{graphicx}
\usepackage{epstopdf}
\usepackage{cancel}
\usepackage{tikz}
\usepackage{algpseudocode}
\usepackage{units}
\usepackage{hyperref}
\usetikzlibrary{shapes,arrows,snakes}
\usetikzlibrary{calc}
\usepackage[utf8x]{inputenc}

\DeclareMathAlphabet{\mathbit}{OML}{cmr}{bx}{it}
\DeclareMathAlphabet{\mathsf}{OT1}{cmss}{m}{n}
\DeclareMathAlphabet{\mathTXf}{OT1}{cmss}{bx}{it}

\DeclareMathOperator{\GDoF}{GDoF}

\DeclareMathOperator{\CN}{\mathcal{N}_{\mathbb{C}}}

\newcommand{\APZF}{\mathrm{APZF}}

\DeclareMathOperator{\SNR}{SNR}

\newcommand{\bH}{\mathbf{H}}

\newcommand{\bh}{\bm{h}}

\newcommand{\bt}{\bm{t}}

\newcommand{\LB}{\left(}
\newcommand{\RB}{\right)}

\newcommand*{\dotleq}{\mathrel{\dot{\leq}}}
\newcommand*{\dotgeq}{\mathrel{\dot{\geq}}}

\newcommand{\E}{{\mathbb{E}}}

\newcommand{\He}{{{\mathrm{H}}}}

\newcommand{\BC}{{\text{BC}}}

\newcommand{\xv}{\mathbf{x}}

\newcommand{\Cb}{{{\mathbb{C}}}}

\newcommand{\Herm}{{{\mathrm{H}}}}

\newcommand{\mformtab}{{{\ \ \ \ \ \ \ }}}

\newcommand{\normv}[1]{\left\|#1\right\|_2}
\newcommand{\abs}[1]{\left|#1\right|}

\theoremstyle{example}
\theoremstyle{assumption}
\newtheorem{theorem}{Theorem}

\newtheorem{lemma}{Lemma}
\newtheorem{corollary}{Corollary}

\newtheorem{remark}{Remark}

\begin{document}

\title{Generalized Degrees-of-Freedom of the 2-User Case MISO Broadcast Channel with Distributed CSIT} 

\author{
      \IEEEauthorblockN{Antonio Bazco\IEEEauthorrefmark{1}\IEEEauthorrefmark{2}, Paul de Kerret\IEEEauthorrefmark{2}, David Gesbert\IEEEauthorrefmark{2}, Nicolas Gresset\IEEEauthorrefmark{1}
}\\
\IEEEauthorblockA{
      \IEEEauthorrefmark{1} Mitsubishi Electric R\&D Centre Europe (MERCE), Rennes, France\\
      % Email: \{a.bazconogueras,n.gresset\}@fr.merce.mee.com \\
      \IEEEauthorrefmark{2} Communication Systems Department, EURECOM\\
      % Email: \{bazco,dekerret,gesbert\}@eurecom.fr 
}} 
\maketitle
\begin{abstract}
This work\footnote{D. Gesbert and P. de Kerret are supported by the European Research Council under the European Union's Horizon 2020 research and innovation program (Agreement no. 670896).} analyses the Generalized Degrees-of-Freedom (GDoF) of the 2-User Multiple-Input Single-Output (MISO) Broadcast Channel (BC) in the so-called Distributed CSIT regime, with application to decentralized wireless networks. This regime differs from the classical limited CSIT one in that the CSIT is not just noisy but also imperfectly shared across the transmitters (TXs). 
Hence, each TX precodes data on the basis of  local CSIT and statistical quality information at other TXs. We derive  the GDoF result and obtain the surprising outcome that by specific accounting of the pathloss information, it becomes possible for the decentralized precoded network to reach the same performance as a genie-aided centralized network where the central node has obtained the estimates of both TXs. The key ingredient in the scheme is the so-called Active-Passive Zero-Forcing (AP-ZF) precoding, which lets the precoder design adapt optimally with respect to different local CSIT qualities available at different TXs.
 
%We use as genie-aided configuration the so-called Centralized CSIT setting, where the imperfect CSIT is \emph{perfectly shared} among all the TXs. Interestingly, we show that the optimal GDoF of the Centralized CSIT setting is achieved using the Active-Passive Zero-Forcing scheme described in \cite{dekerret2016_ISIT_journal}.  
 
 % ii) decodes symbols of any order 
% This work \footnote{D. Gesbert and P. de Kerret are supported by the European Research Council under the European Union's Horizon 2020 research and innovation program (Agreement no. 670896). P. Elia is supported by the ANR project ECOLOGICAL-BITS-AND-FLOPS.} derives the optimal Degrees-of-Freedom (DoF) of the $K$-User MISO Broadcast Channel (BC) with delayed Channel-State Information at the Transmitter (CSIT) and with additional current CSIT where the channel estimation error scales in~$P^{-\alpha}$ for $\alpha\in[0,1]$. The optimal sum DoF takes the simple form $(1-\alpha) K/H_K+\alpha K$ where $H_K\triangleq \sum_{k=1}^K \frac{1}{k}$. This optimal performance is the result of a novel scheme which deviates from existing efforts as it i) digitally combines interference, ii) decodes symbols of any order in the MAT alignment, and iii) utilizes a hierarchical quantizer whose output is distributed across rounds in a way that minimizes unwanted interference. These jointly deliver, for the first time, the elusive simultaneous scaling of both MAT and ZF DoF gains.% ii) decodes symbols of any order 
\end{abstract}

% %%%%%%%%%%%%%%%%%%%%%%%%%%%%%%%%%%%%%%%%%%%%%%%%%%%%%
% %%%%%%%%%%%%%%%%%%%%%%%%%%%%%%%%%%%%%%%%%%%%%%%%%%%%%
% %%%%%%%%%%%%%%%%%%%%%%%%%%%%%%%%%%%%%%%%%%%%%%%%%%%%%
% %%%%%%%%%%%%%%%%%%%%%%%%%%%%%%%%%%%%%%%%%%%%%%%%%%%%%
 \section{Introduction}\label{se:intro}
% %%%%%%%%%%%%%%%%%%%%%%%%%%%%%%%%%%%%%%%%%%%%%%%%%%%%%
% %%%%%%%%%%%%%%%%%%%%%%%%%%%%%%%%%%%%%%%%%%%%%%%%%%%%%
% %%%%%%%%%%%%%%%%%%%%%%%%%%%%%%%%%%%%%%%%%%%%%%%%%%%%%
% %%%%%%%%%%%%%%%%%%%%%%%%%%%%%%%%%%%%%%%%%%%%%%%%%%%%%
Simultaneous transmission between multiple-antennas TXs %(or multiple single-antenna TXs) 
towards different receivers (RXs) has been widely studied, typically assuming a \emph{Centralized CSIT} setting, where only one channel estimate, possibly a noisy one, is used for calculating the precoding coefficients\cite{Jindal2006},\cite{Caire2007}. This can also model a joint transmission from different non-colocated TXs in the case where the CSIT is \emph{perfectly shared} among the TXs over a so-called ideal Cloud Radio Access Network (C-RAN) \cite{Bangerter2014}.
% This can also model a joint transmission from different non-colocated TXs in the case where the CSIT is \emph{perfectly shared} among the TXs over a so-called ideal Cloud Radio Access Network (C-RAN).

However, future wireless network topologies will also include heterogeneous scenarios, with a variety of devices, such as user terminals, drone-enabled relays, pico base stations, etc., seeking to cooperate for transmission despite the lack of an ideal backhaul linking them.  Other scenarios featuring existing backhaul links may favor local processing over centralized one in order to meet  the tight latency constraints derived from 5G and tactile internet applications \cite{Simsek2016}. In these cases, a full CSI sharing across TXs is not always desired, and there is a need for robust processing on the basis of locally available CSI. 

In this paper, we formalize this scenario under the \emph{Distributed CSIT} label, which refers to each TX being endowed with its own version of the multi-TX multi-user channel state matrix, with possibly different qualities. While it was suggested in the past literature that Distributed CSIT scenarios can severely impact on performance in comparison with classical limited-yet-centralized CSIT ones \cite{dekerret2012_TIT}, a crucial and interesting problem is how TXs can cooperatively combat the lack of mutual CSI consistency in order to reduce the gap with respect to the centralized system performance. 

Several works have focused on this Distributed CSIT setting \cite{Zakhour2010a}, e.g., analyzing Interference Alignment performance \cite{dekerret2014_TWC} or studying the Regularized Zero-Forcing performance in the large system limit \cite{Li2015}. However, many of the issues and challenges introduced by this setting are still open problems. 
It has been shown in \cite{dekerret2016_ISIT_journal} that for the $2$-user MISO BC the Distributed CSIT setting achieves the Degrees-of-Freedom (DoF) of the Centralized CSIT setting. Having instantaneous and imperfect CSIT, with an error power scaling as $P^{-\alpha},\ \alpha\in[0,1]$, it is possible to achieve a DoF of $1+\alpha$. This optimal DoF is reached due to a new asymmetrical precoding scheme, so-called Active-Passive Zero-Forcing (AP-ZF), where the most informed TX is able to resolve the error created by the less informed one.

Nevertheless, the DoF is a limited figure of merit, since it does not take into account the differences between channel strengths. In order to study the impact of the network topology, the Generalized DoF (GDoF) concept was introduced in \cite{Etkin2008}. $\GDoF$ approach offers an intermediate step towards finite and constant gap analysis \cite{Davoodi2015_globecom}, modeling the pathlosses through a dependence in $P$\cite{dekerret2014_TIT}. %, offering a finer characterization than the DoF perspective. 
In \cite{Davoodi2016_ISIT_journal_BC} the $\GDoF$ for K-user Symmetric MISO BC with Centralized CSIT has been characterized, and it has been shown that for the $2$-user case the $\GDoF$ only depends on the worst CSIT accuracy towards each RX.

In this work our key contributions are three fold:  First we establish the $\GDoF$ performance of the $2$-user MISO BC under Distributed CSIT for the case where one TX has better CSI quality for all the links. Second we propose a scheme achieving the $\GDoF$, built on the principle of AP-ZF precoding, which is based on the idea that each TX should precode data according to the quality with which it sees CSI. 
Third we show that accounting for pathloss difference in the multi-user channels, the decentralized network can reach the same performance as a genie-aided centralized network where the best CSI estimate is shared. 

%In this work our key contributions are three fold:  First we establish the $\GDoF$ performance of the $2$-user MISO BC under Distributed CSIT. Second we propose a scheme achieving the $\GDoF$ which is based on the idea that each TX should precode data according to the quality with which it sees CSI. This scheme is built on the principle of AP-ZF precoding which is detailed later. Third we show that accounting for pathloss difference in the multi-user channels, the decentralized network can reach the same performance as a genie-aided centralized network where the best TX (i.e. the one with best CSI quality) is allowed to communicate its information at other points of the networks. %as a genie-aided centralized network where all the CSIT information is shared.  
%This surprising result is also validated using Monte Carlo system simulations. 
%%  as a genie-aided centralized network where the best TX (i.e. the one with best CSI quality) is allowed to communicate its information at other points of the networks. 
%Due to space constraint, some of the proofs are sketched, while full proofs for any arbitrary topology can be found in the extended version \cite{Bazco2017_ISIT_journal}.

% GDoF has been used in many works \cite{Tuninetti2007,Wu2007,Mohapatra2012,Vaze2012b,Chaaban2012}, and it has been shown as a interesting approach since optimal achievable schemes for GDoF analysis also achieve within a constant number of bits of capacity \cite{Huang2008}\cite{Bresler2008}.

\emph{Notations:} $\doteq$ denotes the \textit{exponential equality}, i.e., $f(P)\doteq P^{\beta}$ denotes $\lim_{P\rightarrow\infty} \frac{\log\left(f\left(P\right)\right)}{\log\left(P\right)} =\beta$. The \textit{exponential inequalities} $\dotleq$ and $\dotgeq$ are defined in the same manner. $\|\mathbf{A}\|_F$ denotes the Frobenius norm of the matrix $\mathbf{A}$. $\normv{\bt}$ denotes the ${L}^2$-norm of the vector $\bt$, and $\abs{x}$ is the absolute value of the scalar $x$. We define $\bar{i} \triangleq i \pmod 2 + 1$ for $i,\bar{i}\in\{1,2\}$. Being $x$ a number, we define
\begin{align}
		(x)^+ \triangleq \max(x, 0).
\end{align}

\section{System Model}\label{se:SM}
%%%%%%%%%%%%%%%%%%%%%%%%%%%%%%%%%%%%%%%%%%%%%%%%%%%%%
%%%%%%%%%%%%%%%%%%%%%%%%%%%%%%%%%%%%%%%%%%%%%%%%%%%%%
%%%%%%%%%%%%%%%%%%%%%%%%%%%%%%%%%%%%%%%%%%%%%%%%%%%%%
%%%%%%%%%%%%%%%%%%%%%%%%%%%%%%%%%%%%%%%%%%%%%%%%%%%%%

%%%%%%%%%%%%%%%%%%%%%%%%%%%%%%%%%%%%%%%%%%%%%%%%%%%%%
%%%%%%%%%%%%%%%%%%%%%%%%%%%%%%%%%%%%%%%%%%%%%%%%%%%%%
\subsection{$2$-User MISO BC Transmission Model} 
%%%%%%%%%%%%%%%%%%%%%%%%%%%%%%%%%%%%%%%%%%%%%%%%%%%%%
%%%%%%%%%%%%%%%%%%%%%%%%%%%%%%%%%%%%%%%%%%%%%%%%%%%%%
This work considers a communication system where $2$ single-antenna TXs jointly serve $2$ single-antenna RXs over a MISO BC. We assume that the RXs have perfect, instantaneous CSI.
The signal received at RX~$i$ is written as
\begin{align}
			y_i=\bh^\He_i\xv+z_i,
			\label{eq:SM_1}
\end{align}
where $\bh^\He_i\in \Cb^{1\times 2}$ is the channel to user~$i$ and $z_i\in \mathbb{C}$ is the additive Gaussian noise at RX~$i$, distributed in an  independently and identically distributed (i.i.d.) manner as~$\CN(0,1)$. 
$\xv\in \Cb^{2\times 1}$ is the multi-TX transmitted multi-user signal which fulfills the power constraint 
\begin{align}
			\normv{\xv}^2 \doteq P.
\end{align}
$\xv$ is generated from the information symbols $s_i$, which are assumed to be distributed in an i.i.d. manner as~$\CN(0,1)$.
The channel is assumed to be drawn from a continuous ergodic distribution such that all the channel matrices and all their sub-matrices are almost surely full rank \cite{dekerret2011_ISIT}\cite{Davoodi2016}.
  
The relative strength of the elements of the channel matrix $\bH \triangleq [\bh_1,\bh_2]^\Herm$ is modeled as a function of $P$. Given $P$, the nominal $\SNR$ for the scenario without pathloss, it holds that
%The elements of the channel matrix $\bH \triangleq [\bh_1,\bh_2]^\Herm$ enclose the pathlosses that do not vanish for $\SNR\rightarrow\infty$. The relative coefficients of $\SNR$ are denoted by $\gamma_{i,k}\in \left[0,1\right]$. Given $P$, the nominal $\SNR$ for the scenario without pathloss, it holds that
	\begin{align}
 		 \abs{\bH_{i,k}} \doteq \sqrt{P^{\gamma_{i,k}-1}}, \mformtab \forall i,k\in\{1,2\}. \label{eq:channel_loss}
  \end{align}  
 where $\gamma_{i,k}\in \left[0,1\right]$.        
\begin{remark}
For $\gamma_{i,k}=1,\forall i,k\in\{1,2\}$, we recover the conventional DoF setting, while choosing $\gamma_{i,i}=1$, $\gamma_{i,k}=0$, $\forall i,k\in\{1,2\}\mid k\neq i$ we recover the results of the non-interfering IC.\qed
\end{remark}
  
\begin{figure}[h]
 \centering
        \includegraphics[trim = 0mm 100mm 0mm 40mm, clip,width=0.6\textwidth]{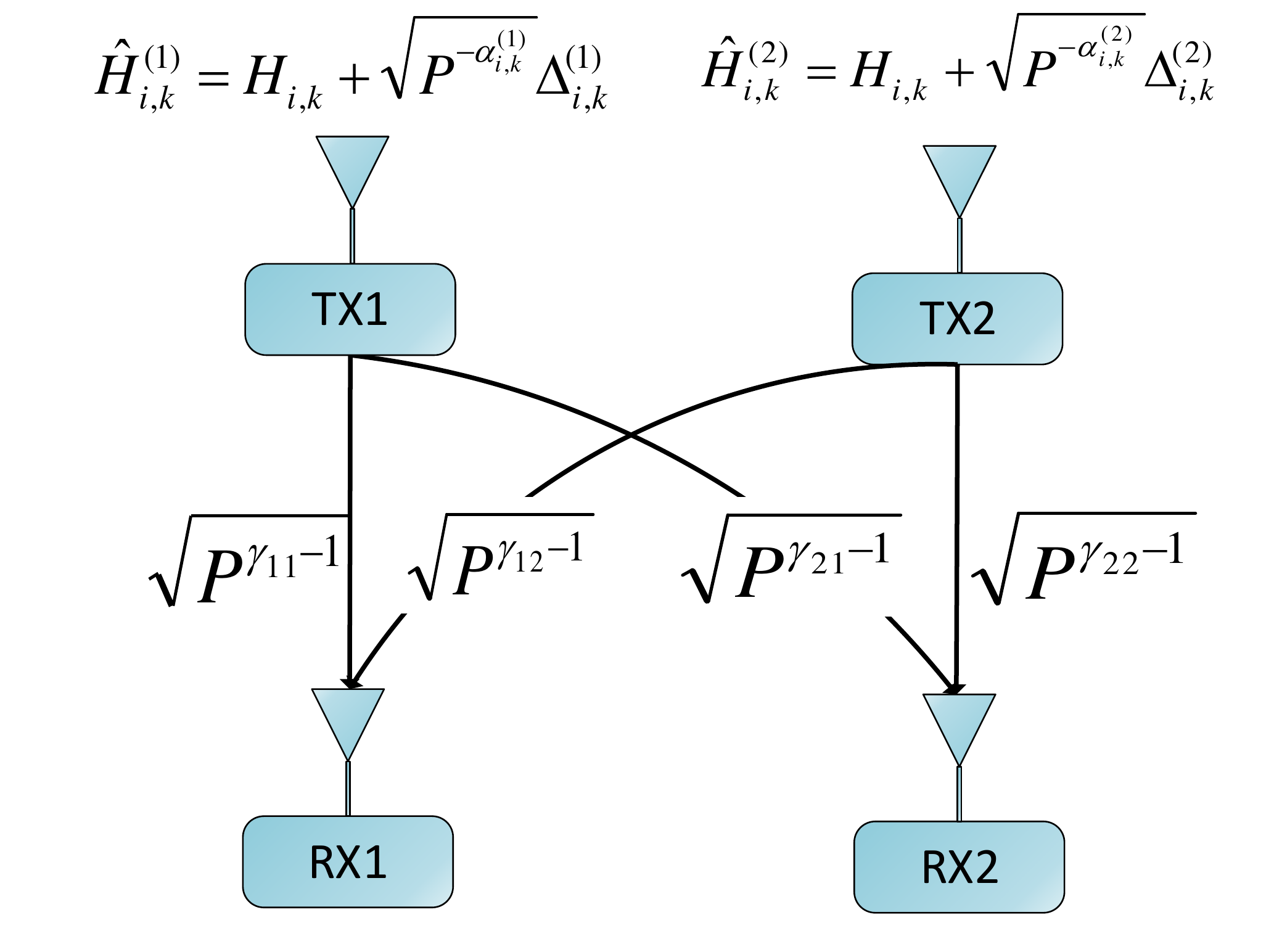}
        \caption{2-user MISO BC System Model with Distributed CSIT.}
		\label{fig:setting}
\end{figure}
The $\GDoF$ approach is a model for the transmission at finite SNR, where the pathlosses are taken into account through a dependence in $P$. For more details, please see \cite{Davoodi2015_globecom},\cite{dekerret2014_TIT}.
\FloatBarrier
%%%%%%%%%%%%%%%%%%%%%%%%%%%%%%%%%%%%%%%%%%%%%%%%%%%%%
%%%%%%%%%%%%%%%%%%%%%%%%%%%%%%%%%%%%%%%%%%%%%%%%%%%%%
\subsection{Distributed CSIT Model}
%%%%%%%%%%%%%%%%%%%%%%%%%%%%%%%%%%%%%%%%%%%%%%%%%%%%%
%%%%%%%%%%%%%%%%%%%%%%%%%%%%%%%%%%%%%%%%%%%%%%%%%%%%%
In that Distributed CSIT setting \cite{dekerret2011_ISIT}, each TX receives a different estimate of the channel, with possibly different accuracies. The CSI uncertainty at the TX~$j$ is modeled as
			\begin{equation}
						\hat{\bH}^{(j)}_{i,k} \triangleq \bH_{i,k}+\sqrt{P^{-\alpha^{(j)}_{i,k}}}\bm{\Delta}^{(j)}_{i,k},  \mformtab \forall j \in \{1,2\},
						\label{eq:distributed model}
			\end{equation} 
where $\bm{\Delta}^{(j)}_{i,k}$ are independent random variables with zero mean and bounded covariance matrix satisfying $\abs{\bm{\Delta}^{(j)}_{i,k}} \doteq \sqrt{P^{\gamma_{i,k}-1}}$, $\forall i,k$.
%It holds as in~\eqref{eq:channel_loss} that
% \begin{align}
%  		 \bm{\Delta}^{(j)} \doteq \bm{\Gamma}.\label{eq:channel_loss_delta}
%   \end{align}
The \emph{CSIT quality exponent} at TX~$j$ is denoted as $\alpha^{(j)}_{i,k} \in [0,\gamma_{i,k}]$  and it is used to parameterize the accuracy of the current CSIT. 
Note that from a GDoF perspective, $\alpha^{(j)}_{i,k}$ can be limited to $\alpha^{(j)}_{i,k} \in [0,\gamma_{i,k}]$. This comes from the fact that, in terms of GDoF, an estimation with error scaling as~$P^{-1}$ can be intuitively understood as being perfect while an estimation with error scaling as~$P^{0}$ is not helpful \cite{Davoodi2015_globecom}. We assume that TX~$1$ is the most informed TX throughout the work, i.e., 
\begin{align}
	1 \geq \alpha^{(1)}_{i,k} \geq \alpha^{(2)}_{i,k} \geq 0 . \label{eq:alphas_order}
\end{align}
%This restriction ensures the optimality of AP-ZF precoder. The extension to arbitrary values of $\{\alpha_{i,k}^{(j)}\}$ is currently under investigation in the group. 
The more-informed TX assumption is key to the optimality of AP-ZF. Extending the results to the arbitrary CSIT regime is an interesting research topic currently under investigation.

In addition, we assume that the conditional probability density functions verify that
    \begin{align}
		\E \big[ \|\bH_{i,k} - \E[\bH_{i,k}|\hat{\bH}^{(1)}_{i,k},\hat{\bH}^{(2)}_{i,k}]\| \big] = O(P^{\max_{j\in\{1,2\}}\alpha^{(j)}_{i,k}}).\label{eq:cond_prob}
    \end{align} 
This technical condition extends the statement from \cite{dekerret2016_ISIT_journal},~\cite{Davoodi2016} and it is satisfied for the usually assumed distributions.

%%%%%%%%%%%%%%%%%%%%%%%%%%%%%%%%%%%%%%%%%%%%%%%%%%%%%
\subsection{Generalized Degrees-of-Freedom Analysis}
%%%%%%%%%%%%%%%%%%%%%%%%%%%%%%%%%%%%%%%%%%%%%%%%%%%%%
			 The optimal sum $\GDoF$ in the MISO BC scenario with imperfect current CSIT is defined as \cite{Etkin2008}
			\begin{equation}
			\GDoF^{\star} \triangleq \lim_{P\rightarrow \infty}\frac{\mathcal{C}(P)}{\log_2(P)},
			\label{eq:sysmod_gdof}
			\end{equation}
where $\mathcal{C}(P)$ denotes the sum capacity \cite{Cover2006} of the MISO BC studied.

%%%%%%%%%%%%%%%%%%%%%%%%%%%%%%%%%%%%%%%%%%%%%%%%%%%%%
%%%%%%%%%%%%%%%%%%%%%%%%%%%%%%%%%%%%%%%%%%%%%%%%%%%%%
%%%%%%%%%%%%%%%%%%%%%%%%%%%%%%%%%%%%%%%%%%%%%%%%%%%%%
%%%%%%%%%%%%%%%%%%%%%%%%%%%%%%%%%%%%%%%%%%%%%%%%%%%%%
%\section{Preliminary: Discussion of the Centralized CSIT Case}\label{se:prev}
\section{Preliminary: Results of the Centralized CSIT Case}\label{se:prev}
%%%%%%%%%%%%%%%%%%%%%%%%%%%%%%%%%%%%%%%%%%%%%%%%%%%%%
%%%%%%%%%%%%%%%%%%%%%%%%%%%%%%%%%%%%%%%%%%%%%%%%%%%%%
%%%%%%%%%%%%%%%%%%%%%%%%%%%%%%%%%%%%%%%%%%%%%%%%%%%%%
%%%%%%%%%%%%%%%%%%%%%%%%%%%%%%%%%%%%%%%%%%%%%%%%%%%%%

%%%%%%%%%%%%%%%%%%%%%%%%%%%%%%%%%%%%%%%%%%%%%%%%%%%%%
%%%%%%%%%%%%%%%%%%%%%%%%%%%%%%%%%%%%%%%%%%%%%%%%%%%%%
\subsection{Centralized CSIT Model}\label{subse:prev_model}
%%%%%%%%%%%%%%%%%%%%%%%%%%%%%%%%%%%%%%%%%%%%%%%%%%%%%
%%%%%%%%%%%%%%%%%%%%%%%%%%%%%%%%%%%%%%%%%%%%%%%%%%%%% 
We now focus on the Centralized CSIT configuration. This setting is useful as point of reference for the analysis of the effect of the discrepancies between TXs that appear in the Distributed CSIT setting.
In this centralized setting all the TXs share the exact same, potentially imperfect, channel estimate. Hence, there is a single channel estimate such that we can remove the TX index and consider simply $\hat{\bH}$. Thus, the CSI uncertainty at the TXs is modeled as
 			\begin{equation}
             	\hat{\bH}_{i,k}=\bH_{i,k}+\sqrt{P^{-\alpha_{i,k}}}\bm{\Delta}_{i,k}.					\label{eq:centralized_model}
 			\end{equation}
      %\begin{equation}
					%\hat{\bH}=\bH+\sqrt{P^{-\alpha}}\bm{\Delta}.					\label{eq:centralized_model}
			%\end{equation}
This setting models the scenario where the precoding is done at a central node or where the CSIT information is \emph{perfectly shared} between the TXs. 
%%%%%%%%%%%%%%%%%%%%%%%%%%%%%%%%%%%%%%%%%%%%%%%%%%%%%
%%%%%%%%%%%%%%%%%%%%%%%%%%%%%%%%%%%%%%%%%%%%%%%%%%%%%
\subsection{Generalized Degrees-of-Freedom of the Centralized CSIT Setting}\label{subse:gdof_centr}
%%%%%%%%%%%%%%%%%%%%%%%%%%%%%%%%%%%%%%%%%%%%%%%%%%%%%
%%%%%%%%%%%%%%%%%%%%%%%%%%%%%%%%%%%%%%%%%%%%%%%%%%%%% 
The GDoF of the $2$-user MISO BC with Centralized CSIT has been derived in \cite{Davoodi2016_ISIT_journal_BC}. %thanks to a new outer bound. 
We provide in the following their main result for the setting considered in this work.
\begin{theorem}\cite{Davoodi2016_ISIT_journal_BC} \label{theo:centralized_outerbound}
			In the $2$-user MISO BC with Centralized CSIT the optimal sum $\GDoF$, denoted as $\GDoF^{CCSIT}(\alpha)$, satisfies  
			\begin{align}
				\GDoF^{CCSIT}(\{\alpha_{i,k}\}_{i,k\in\{1,2\}})=\min(D_1,D_2), \label{eq:centralized_gdof_min}
			\end{align}
			where 
			\begin{align}
				D_1  &\triangleq \max \left(\gamma_{1,2}, \gamma_{1,1} \right) + \max \left( (\gamma_{2,1} - \gamma_{1,1} + \min(\alpha_{1,1},\alpha_{1,2})) ^+  ,(\gamma_{2,2} - \gamma_{1,2} + \min(\alpha_{1,1},\alpha_{1,2})) ^+ \right), \nonumber\\
				D_2  &\triangleq \max \left(\gamma_{2,2}, \gamma_{2,1} \right) + \max \left( (\gamma_{1,1} - \gamma_{2,1} + \min(\alpha_{2,1},\alpha_{2,2})) ^+  ,(\gamma_{1,2} - \gamma_{2,2} + \min(\alpha_{2,1},\alpha_{2,2})) ^+ \right).\nonumber
			\end{align}				
\end{theorem}
Interestingly, depending the network geometry the pathlosses can be either advantageous (since they reduce the interference power received) or detrimental (since they reduce the intended signal power received in the same level that the interference). Moreover, the $\GDoF$ performance is only depends on the weakest CSIT parameter for each receiver. %If the maximum value in the second term is $0$, the optimal GDoF is achieved by transmitting only to one RX
For ease of notation, we introduce the short-hand notations 
    \begin{align}
		\alpha_1 & \triangleq \min \left(\alpha_{1,1},\alpha_{1,2}\right), \\
		\alpha_2 & \triangleq \min \left(\alpha_{2,1},\alpha_{2,2}\right). 
	\end{align}
%For ease of notation, we define
    %\begin{align}
		%\alpha_1 & \triangleq \min \left(\alpha_{1,1},\alpha_{1,2}\right), \\
		%\alpha_2 & \triangleq \min \left(\alpha_{2,1},\alpha_{2,2}\right). 
	%\end{align}  
\begin{remark} 
		This optimal sum GDoF is achieved by superposition coding and ZF precoding \cite{Jindal2006}\cite{Hao2015a}. \qed
\end{remark}
%and therefore the $\GDoF$ expression from Theorem~\ref{theo:centralized_outerbound} can be written as
%\begin{align}
				%\GDoF^{CCSIT}(\{\alpha_{i,k}\}_{i,k\in\{1,2\}})&=\min\left( \max \left(\gamma_{1,2}, \gamma_{1,1} \right) + \max \left( (\gamma_{2,1} - \gamma_{1,1} + \alpha_1) ^+  ,(\gamma_{2,2} - \gamma_{1,2} + \alpha_1) ^+ \right), \right.\nonumber\\
				%&\qquad \left. \max \left(\gamma_{2,2}, \gamma_{2,1} \right) + \max \left( (\gamma_{1,1} - \gamma_{2,1} + \alpha_2) ^+  ,(\gamma_{1,2} - \gamma_{2,2} + \alpha_2) ^+ \right).\nonumber\label{eq:centralized_gdof_min_short}
				%\right).	
%\end{align}		
%%%%%%%%%%%%%%%%%%%%%%%%%%%%%%%%%%%%%%%%%%%%%%%%%%%%%%%%%%%%%%%%%%%
\subsection{Centralized Outerbound}\label{subse:centr_outer}
%%%%%%%%%%%%%%%%%%%%%%%%%%%%%%%%%%%%%%%%%%%%%%%%%%%%%%%%%%%%%%%%%%%
As a first step toward characterizing the GDoF, we extend the centralized outerbound derived in \cite{dekerret2016_ISIT_journal} for the conventional DoF to the $\GDoF$ setting.
\begin{theorem}\label{theo:centralized_outerbound_proof}
In the $2$-user MISO BC with Distributed CSIT, the optimal $\GDoF$ is upperbounded by the $\GDoF$ of a Centralized CSIT scenario in which all the TXs estimations are perfectly shared. Concretely, it holds that
	\begin{align}
								\GDoF^{DCSIT}(\{\alpha^{(j)}_{i,k}\}_{i,j,k\in\{1,2\}}) \leq \GDoF^{CCSIT}(\{\max_{j\in\{1,2\}}\alpha^{(j)}_{i,k}\}_{i,k\in\{1,2\}}).
	\end{align}
\end{theorem}
\begin{proof} 
Assuming a genie-aided model where all the TXs share \emph{perfectly} his local estimation, and denoting the total available CSIT as $\mathcal{H} \triangleq \{\hat{\bH}^{(1)},\hat{\bH}^{(2)}\}$, it holds from \eqref{eq:cond_prob} that it is possible to apply the centralized outerbound in \cite{Davoodi2016_ISIT_journal_BC}.
\end{proof}
This centralized genie-aided model satisfies that
\begin{align}
	\alpha_{i,k} \triangleq \max_{j\in\{1,2\}}(\alpha^{(j)}_{i,k}), \qquad \forall i,k\in\{1,2\}.
\end{align}

%%%%%%%%%%%%%%%%%%%%%%%%%%%%%%%%%%%%%%%%%%%%%%%%%%%%%
%%%%%%%%%%%%%%%%%%%%%%%%%%%%%%%%%%%%%%%%%%%%%%%%%%%%%
%%%%%%%%%%%%%%%%%%%%%%%%%%%%%%%%%%%%%%%%%%%%%%%%%%%%%
%%%%%%%%%%%%%%%%%%%%%%%%%%%%%%%%%%%%%%%%%%%%%%%%%%%%%
\section{Main Results}\label{se:main}
%%%%%%%%%%%%%%%%%%%%%%%%%%%%%%%%%%%%%%%%%%%%%%%%%%%%%
%%%%%%%%%%%%%%%%%%%%%%%%%%%%%%%%%%%%%%%%%%%%%%%%%%%%%
%%%%%%%%%%%%%%%%%%%%%%%%%%%%%%%%%%%%%%%%%%%%%%%%%%%%%
%%%%%%%%%%%%%%%%%%%%%%%%%%%%%%%%%%%%%%%%%%%%%%%%%%%%%

We can now state our main result.
% The main result of this work is the sum GDoF for the Distributed CSIT setting stated in the next Theorem.
\begin{theorem}\label{theo:distributed_gdof}
In the $2$-user MISO BC with Distributed CSIT, the optimal sum $\GDoF$ is given by
	\begin{align}
		\GDoF^{DCSIT}(\{\alpha^{(j)}_{i,k}\}_{i,j,k\in\{1,2\}}) = \GDoF^{CCSIT}(\{\max_{j\in\{1,2\}}\alpha^{(j)}_{i,k}\}_{i,k\in\{1,2\}}).
        \label{eq:dist_cent}
	\end{align}	
\end{theorem}     
% \begin{theorem}
% In the $2$-user MISO BC with Distributed CSIT, the optimal sum $\GDoF$ is given by
% 	\begin{align}
% 		\GDoF^{DCSIT}(\{\alpha^{(1)}_{i,k}\},\{\alpha^{(2)}_{i,k}\}) = \GDoF^{CCSIT}(\{\max_{j\in\{1,2\}}\alpha^{(j)}_{i,k}\}).
%         \label{eq:dist_cent}
% 	\end{align}	\label{theo:distributed_gdof}
% \end{theorem}       
\begin{proof}
The outer bound comes directly from the centralized outer bound presented above and the main contribution is to prove the achievability. This is done by showing that AP-ZF, which is optimal in terms of DoF, is also optimal in terms of Generalized DoF. This requires to prove first some intermediate results in Section~\ref{se:apzf_prem} before turning to the proof. In Section~\ref{se:achievability_simple} a insightful case is shown to get a main insight of the scheme, while the general proof of the achievability is shown in the appendix.
\end{proof}

Surprisingly, even when the most informed TX has only weak links, the system behaves as if both TXs shared the best CSIT estimate, which can be rather counter-intuitive.
Since Theorem~\ref{theo:centralized_outerbound} shows that the $\GDoF$ only depends on the weakest CSIT parameter for each RX, we define the distributed counterparts of $\alpha_1,\alpha_2$ as 
    %\begin{align}
		%\alpha^{(j)}_1 & \triangleq \min \left(\alpha^{(j)}_{1,1},\alpha^{(j)}_{1,2}\right),\\
		%\alpha^{(j)}_2 & \triangleq \min \left(\alpha^{(j)}_{2,1},\alpha^{(j)}_{2,2}\right),
	%\end{align}  
	  \begin{align}
		%\alpha^{(j)}_1 & \triangleq \min \left(\alpha^{(j)}_{1,1},\alpha^{(j)}_{1,2}\right), \qqad \forall j \in \{1,2\},\\
		%\alpha^{(j)}_2 & \triangleq \min \left(\alpha^{(j)}_{2,1},\alpha^{(j)}_{2,2}\right), \qqad \forall j \in \{1,2\},		
		\alpha^{(j)}_1 & \triangleq \min_{k\in\{1,2\}} \alpha^{(j)}_{1,k}, \qquad \forall j \in \{1,2\},\\
		\alpha^{(j)}_2 & \triangleq \min_{k\in\{1,2\}} \alpha^{(j)}_{2,k}, \qquad \forall j \in \{1,2\}.
	\end{align}  
% The outerbound comes directly from the centralized outerbound presented above and the main contribution is to prove the achieveability. This is done by showing that AP-ZF, which is optimal in terms of DoF, is also optimal in terms of generalized DoF. This requires to prove first some intermediate results in Section 3.1 before turning to the main proof in Section 4.
%$\forall j \in \{1,2\}$. 
The main consequences of the $\GDoF$ model for the Distributed CSIT setting are stated in the following.

%%%%%%%%%%%%%%%%%%%%%%%%%%%%%%%%%%%%%%%%%%%%%%%%%%%%%
%%%%%%%%%%%%%%%%%%%%%%%%%%%%%%%%%%%%%%%%%%%%%%%%%%%%%
%%%%%%%%%%%%%%%%%%%%%%%%%%%%%%%%%%%%%%%%%%%%%%%%%%%%%
%%%%%%%%%%%%%%%%%%%%%%%%%%%%%%%%%%%%%%%%%%%%%%%%%%%%%
\section{Preliminaries: Analysis of APZF Precoding}\label{se:apzf_prem}
%%%%%%%%%%%%%%%%%%%%%%%%%%%%%%%%%%%%%%%%%%%%%%%%%%%%%
%%%%%%%%%%%%%%%%%%%%%%%%%%%%%%%%%%%%%%%%%%%%%%%%%%%%%
%%%%%%%%%%%%%%%%%%%%%%%%%%%%%%%%%%%%%%%%%%%%%%%%%%%%%
%%%%%%%%%%%%%%%%%%%%%%%%%%%%%%%%%%%%%%%%%%%%%%%%%%%%%
We firstly characterize the $2$-user AP-ZF precoder behaviour in the GDoF model. For that, the power consumption at each TX is stated for any possible network topology, and from that result the intended signal received power and the remaining interference power are presented. 

%%%%%%%%%%%%%%%%%%%%%%%%%%%%%%%%%%%%%%%%%%%%%%%%%%%%%
%%%%%%%%%%%%%%%%%%%%%%%%%%%%%%%%%%%%%%%%%%%%%%%%%%%%%
\subsection{AP-ZF Precoder for the $2$-user Setting}\label{subse:apzf}
%%%%%%%%%%%%%%%%%%%%%%%%%%%%%%%%%%%%%%%%%%%%%%%%%%%%%
%%%%%%%%%%%%%%%%%%%%%%%%%%%%%%%%%%%%%%%%%%%%%%%%%%%%%
For the sake of completeness, the AP-ZF precoder first introduced in \cite{dekerret2016_ISIT_journal} is briefly presented, as it is a key component of the proposed transmission scheme. The core feature of this precoder is the uneven precoding that allows the most informed TX to neutralize the interference generated by the other TX.

Let RX~$i$ be the intended RX and RX~$\bar{i}$ be the interfered RX. As TX~$1$ is the most informed TX, the AP-ZF beamformer is given by
\begin{align}
t^{(2)}_i &\triangleq c_P,\label{eq:apzf_tx_pas}	\\
t^{(1)}_i &\triangleq -\hat{h}^{(1)}_{\bar{i},1}\left(\abs{\hat{h}^{(1)}_{\bar{i},1}}^2 + \frac{1}{P}\right)^{-1}\hat{h}^{(1)\Herm}_{\bar{i},2} t^{(2)}_i, \label{eq:apzf_tx_act}	
\end{align}
where $c_P$ is a constant that can be made dependent on $P$. 
Therefore, for the transmission towards a certain RX, the less informed TX, so-called \emph{passive TX}, selects as fixed precoding coefficient known by both TXs and thus it does not use its own CSIT information. On the other hand, the most informed TX, so-called \emph{active TX}, selects the precoder coefficient that generates a received signal at the interfered RX with the opposite phase of the one that comes from the other TX. 
It can easily be seen that, as a consequence of that precoding scheme, the interference power received at RX~$\bar{i}$ is decreased by a factor $P^{-\alpha^{(1)}_{\bar{i}}}$\cite{dekerret2016_ISIT_journal}. 

%%%%%%%%%%%%%%%%%%%%%%%%%%%%%%%%%%%%%%%%%%%%%%%%%%%%%
%%%%%%%%%%%%%%%%%%%%%%%%%%%%%%%%%%%%%%%%%%%%%%%%%%%%%
\subsection{Power Consumption}\label{subse:lemmas}
%%%%%%%%%%%%%%%%%%%%%%%%%%%%%%%%%%%%%%%%%%%%%%%%%%%%%
%%%%%%%%%%%%%%%%%%%%%%%%%%%%%%%%%%%%%%%%%%%%%%%%%%%%%

The main impact of the $\GDoF$ model comes from the power normalization at the TXs, as shown in the following lemma.

%The key factor of the achievability is the received power harmonization, since the transmission scheme is based on interference  neutralization, and thus the received interference power from each transmitter must be the same. We show in the following Lemma that the APZF precoder adapts by itself the power consumption at each TX in order to accomplish this condition.

%\begin{lemma}\label{lem:power_scaling}
%In the $2$-user MISO BC, the APZF precoder $\bt^{\APZF}_i$ aimed to RX~$i$ and transmitted with power $||\bt^{\APZF}_i||_2^2 \doteq P^{\tau}$, $\tau \in[0,1]$, with $\ i,\bar{i} \in \{1,2\}$, $i\neq \bar{i}$,  satisfies
                    %\begin{align}
                             %{\abs{t^{(1)}_i}^2} &\doteq P^{1-(\gamma_{\bar{i},1}-\gamma_{\bar{i},2})^+},\label{eq:lemma_eq_1}\\
                             %{\abs{t^{(2)}_i}^2} &\doteq P^{1-(\gamma_{\bar{i},2}-\gamma_{\bar{i},1})^+}. \label{eq:lemma_eq_2}
                     %\end{align}				
%\end{lemma}

\begin{lemma}\label{lem:power_scaling}
In the $2$-user MISO BC, the AP-ZF precoder $\bt^{\APZF}_i$ aimed to RX~$i$ and transmitted with power $||\bt^{\APZF}_i||_2^2 \doteq P^{\tau}$,  $\tau \in[0,1]$, 
%, with $i,\bar{i} \in\{1,2\}, \bar{i}\neq i$, 
satisfies
			\begin{align}
							 {\abs{t^{(1)}_i}^2} &\doteq P^{\tau-(\gamma_{\bar{i},1}-\gamma_{\bar{i},2})^+},\label{eq:lemma_eq_1}\\
							 {\abs{t^{(2)}_i}^2} &\doteq P^{\tau-(\gamma_{\bar{i},2}-\gamma_{\bar{i},1})^+}.\label{eq:lemma_eq_2} 
			 \end{align}				
%By symmetry, this also holds for the AP-ZF precoder $\bt^{\APZF}_2$ aimed to RX~$2$ by switching the user-index.
\end{lemma}

\begin{proof}
Letting the constant coefficient of \eqref{eq:apzf_tx_pas} fulfill
		\begin{equation}
		\begin{aligned}
			\abs{t^{(2)}_i} \doteq \sqrt{P^x} \label{eq:power_scaling_passive},\mformtab x\in [0,1], 
		\end{aligned}	 
		\end{equation}
	it holds from \eqref{eq:apzf_tx_act} that the coefficient designed at TX~$1$ satisfies
		\begin{equation}
		\begin{aligned}
				\abs{t^{(1)}_i} &=	\abs{\hat{h}^{(1)}_{\bar{i},1}} \abs{-\left(\abs{\hat{h}^{(1)}_{\bar{i},1}}^{2} + \frac{1}{P}\right)^{-1}}\abs{\hat{h}^{(1)\Herm}_{\bar{i},2}}\abs{t^{(2)}_i}. \label{eq:precod_2us_def}
		\end{aligned}	 
		\end{equation}	  
By definition (see equation~\eqref{eq:channel_loss}), it also holds 
%From \eqref{eq:channel_loss} it holds that
				\begin{align}
						\abs{\hat{h}^{(1)}_{\bar{i},2}} \doteq \sqrt{P^{\gamma_{\bar{i},2}-1}} \label{eq:passive_doteq},\\
						\abs{\hat{h}^{(1)}_{\bar{i},1}} \doteq \sqrt{P^{\gamma_{\bar{i},1}-1}} \label{eq:active_doteq},
				\end{align}
and then the absolute value satisfies that
		\begin{align}
        	\abs{t^{(1)}_i}   &\doteq   \sqrt{P^{\gamma_{\bar{i},1}-1}}\abs{P^{\gamma_{\bar{i},1}-1}(1   +   P^{-\gamma_{\bar{i},1}})}^{-1}   \sqrt{P^{\gamma_{\bar{i},2}-1}}\sqrt{P^x} \nonumber\\%\label{eq:bounding_inverse} \\ 
                 &\doteq \sqrt{P^{x+(\gamma_{\bar{i},2} - \gamma_{\bar{i},1})}}. \label{eq:power_scaling_active}
		\end{align}
From \eqref{eq:power_scaling_passive} and \eqref{eq:power_scaling_active}, given that the final precoder should have a power of $ \|\bt^{\APZF}_i\|_{2} = \sqrt{P^{\tau}}$, the optimal choices for $x$ are
		\begin{align}
				&\sqrt{P^x} = \begin{cases}
					\sqrt{P^{\tau}} &\text{ if $\gamma_{\bar{i},2}-\gamma_{\bar{i},1} \leq 0$},\\
				 	\sqrt{P^{\tau-(\gamma_{\bar{i},2}-\gamma_{\bar{i},1})}}   &\text{ if $\gamma_{\bar{i},2} -\gamma_{\bar{i},1} > 0$},
				\end{cases}
		\end{align}	
        which concludes the proof.
% 			and thus we obtain the result of Lemma~\ref{lem:power_scaling}.
\end{proof}

\begin{remark}\label{remark:lower_power}
     It can be seen from \eqref{eq:lemma_eq_1}-\eqref{eq:lemma_eq_2} that there is always one TX which reaches the power constraint (i.e. $P^{\tau}$), while at the other TX the power is reduced to $P^{\tau - \abs{\gamma_{2,2}-\gamma_{2,1}}}$. \qed
\end{remark} 
% Lemma~\ref{lem:power_scaling} shows that the transmitted power from each TX might be different in order to adapt the precoding to the network topology. %Then, the less informed TX~$2$ does not use its own CSIT and only the relation between channel strength is needed for normalize correctly the transmission. 
% % The key factor of the achievability is the received power harmonization,
% Since the precoding scheme is based on interference neutralization, the received interference power from each transmitter must be the same. 
% Thus, the transmitted power must compensate the differences in pathloss between the links to a given RX.
Building upon Lemma~\ref{lem:power_scaling}, the following results on the scaling of the received signals are easily obtained from the network topology.

%This Lemma shows that, depending on the network geometry, one of the two TXs consumes less power. This unequal power consumption comes from the unequal strength of the different paths, and leads to a matching of the received interferences, what is a required condition for the interference cancellation. Only the relation between channel strength is needed for normalize correctly the transmission.

%From this consumption power characterization, the intended signal and interference power are obtained straightforwardly. The following corollary come from Lemma~\ref{lem:power_scaling} and the network topology.

\begin{corollary}\label{cor:received}
In the $2$-user MISO BC with Distributed CSIT, transmitting with power $||\bt^{\APZF}_i||_2^2 \doteq P^{\tau}$, the intended signal received power at  RX~$i$, $i\in\{1,2\}$, satisfies
			\begin{equation}
			\begin{aligned}
						 \abs{\bh^\Herm_i{\bt^{\APZF}_i} }^2   &\doteq  P^{\tau-1} \max(P^{\gamma_{i,1} -(\gamma_{\bar{i},1}-\gamma_{\bar{i},2})^+}   , P^{\gamma_{i,2} -(\gamma_{\bar{i},2}-\gamma_{\bar{i},1})^+} ), \label{eq:received_power}
			\end{aligned}
			\end{equation}
while the interference power at the same RX~$i$ from the signal intended to the other RX~$\bar{i}$ satisfies
			\begin{equation}
						\abs{\bh^\He_i{\bt^{\APZF}_{\bar{i}}}}^2 \dotleq P^{\tau-1}P^{\min(\gamma_{i,1},\gamma_{i,2})-\alpha^{(1)}_i}.\label{eq:received_interference}
			\end{equation}
\end{corollary}
As main insight, it is noted that the $\pm(\gamma_{\bar{i},2}-\gamma_{\bar{i},1})$ terms in \eqref{eq:received_power}, as well as the $\min(\gamma_{i,1},\gamma_{i,2})$ term in \eqref{eq:received_interference}, come from the fact that the TX with greater channel strength reduces his power to match the power received from the other TX so as to be able to cancel the interference. 

\begin{proof}
As we are analyzing the $2$-user case, it holds that
			\begin{align}
				\abs{\bh^\Herm_i\bt^{\APZF}_i}^2
                  \doteq \max\Big(\Big|{h^\Herm_{i,1}t^{(1)}_{i}}\Big|^2 ,\Big|{h^\Herm_{i,2}t^{(2)}_{i}}\Big|^2 \Big).		\label{eq:proof_lower}
            \end{align}
From Lemma~\ref{lem:power_scaling}, it holds that the power of both coefficients differs in $P^{|\gamma_{\bar{i},2}-\gamma_{\bar{i},1}|}$. Assuming that the transmitted power scales as $P^\tau$, it holds that
			\begin{align}
						\abs{h^\Herm_{i,1}t^{(1)}_{i}}^2 &=\abs{h^\Herm_{i,1}}^2\abs{t^{(1)}_{i}}^2\\
                        &\doteq P^{\gamma_{i,1}-1} {P^{\tau-(\gamma_{\bar{i},1}-\gamma_{\bar{i},2})^+}},
                        \label{eq:proof_lower_a}
			\end{align}
and, in the same way,
			\begin{align}
									\abs{h^\Herm_{i,2}t^{(2)}_{i}}^2 &=\abs{h^\Herm_{i,2}}^2\abs{t^{(2)}_{i}}^2\\
                        &\doteq P^{\gamma_{i,2}-1} {P^{\tau-(\gamma_{\bar{i},2}-\gamma_{\bar{i},1})^+}}.
                        \label{eq:proof_lower_p}
			\end{align}
			Including \eqref{eq:proof_lower_a} and \eqref{eq:proof_lower_p} in \eqref{eq:proof_lower}, we prove \eqref{eq:received_power}. Focusing now in the proof of the interference power expression in \eqref{eq:received_interference}, it holds that
%where the second factor of \eqref{eq:proof_lower_a} and \eqref{eq:proof_lower_p} comes from the precoder design showed in Lemma~\ref{lem:power_scaling}. Including \eqref{eq:proof_lower_a}-\eqref{eq:proof_lower_p} in \eqref{eq:proof_lower}, we prove the Lemma~\ref{theo:received_power}.  
			\begin{align}
						\abs{\bh^\He_i{\bt^{\APZF}_{\bar{i}}}}^2 & = \abs{h^\Herm_{i,1}t^{(1)}_{\bar{i}} + h^\Herm_{i,2}t^{(2)}_{\bar{i}}}^2. \label{eq:inter_proof1}%\\
						%&=P^{\tau-1}P^{\min(\gamma_{i,1},\gamma_{i,2})-\alpha^{(1)}}.\label{eq:inter_proof1}
			\end{align}
As stated in Remark~\ref{remark:lower_power}, the TX with stronger interfering channel reduces his transmission power in a factor $P^{\abs{\gamma_{\bar{i},2}-\gamma_{\bar{i},1}}}$, and hence the terms in \eqref{eq:inter_proof1} satisfy
			\begin{align}
						\abs{h^\Herm_{i,1}t^{(1)}_{\bar{i}}}^2 = P^{\tau-1+\min({\gamma_{\bar{i},1},\gamma_{\bar{i},2}})},\\
						\abs{h^\Herm_{i,2}t^{(2)}_{\bar{i}}}^2 = P^{\tau-1+\min({\gamma_{\bar{i},1},\gamma_{\bar{i},2}})}. %\label{eq:inter_proof1}%\\
			\end{align}
As it is known that the AP-ZF precoder reduces the interference by a factor $P^{\alpha^{(1)}_i}$ \cite{dekerret2016_ISIT_journal}, it holds that
			\begin{align}
						\abs{\bh^\He_i{\bt^{\APZF}_{\bar{i}}}}^2 & = \abs{h^\Herm_{i,1}t^{(1)}_{\bar{i}} + h^\Herm_{i,2}t^{(2)}_{\bar{i}}}^2 \\
						&=P^{\tau-1}P^{\min(\gamma_{i,1},\gamma_{i,2})-\alpha^{(1)}_i},
			\end{align}
			which concludes the proof.
\end{proof}

%As main insight, it is noted that the $\min(\gamma_{i,1},\gamma_{i,2})$ term comes from the fact that the stronger TX reduces his power to level it off with the signal received from the weaker TX and hence being able to remove the interference. 

\begin{remark}
		The value of $\alpha^{(1)}_i$ make only sense in the interval $[0,\min(\gamma_{i,1},\gamma_{i,2})]$, since a signal scaling in $P^{-1}$ has no impact in terms of DoF/GDoF. \qed %is considered negligible. \qed
\end{remark} 

\section{Achievability in the Parallel Configuration}\label{se:achievability_simple}
%\section{Achievability of Insightful Case}\label{se:achievability_simple}
%%%%%%%%%%%%%%%%%%%%%%%%%%%%%%%%%%%%%%%%%%%%%%%%%%%%%
%%%%%%%%%%%%%%%%%%%%%%%%%%%%%%%%%%%%%%%%%%%%%%%%%%%%%
%%%%%%%%%%%%%%%%%%%%%%%%%%%%%%%%%%%%%%%%%%%%%%%%%%%%%
%%%%%%%%%%%%%%%%%%%%%%%%%%%%%%%%%%%%%%%%%%%%%%%%%%%%%
%In the following, Theorem~\ref{theo:distributed_gdof} is proved for two specific, exemplifying topologies. This allows to convey the main intuition of the proof while avoiding cluttered and heavy notations. The achievability for the general setting is provided in the Appendix.
In the following, Theorem~\ref{theo:distributed_gdof} is proved for one specific topology, which we denote as the \emph{Parallel Configuration}. This simple setting allows to convey the main intuition of the proof while avoiding cluttered and heavy notations. The proof for any possible topology is given in the Appendix. 
%For ease of notation and analysis, we restrict the number of different parameters. Therefore, there are two different channel strengths: The channels strength of the good channels is $\gamma_{i,j}=1$ and the strength of the weak channels is $\gamma_{i,j}=\gamma$.
%We recall that in those cases, $\alpha^{(j)}$ is restricted to 
    %\begin{align}
        %\alpha^{(j)} \in [0,\gamma]. \label{eq:alpha_limit}
    %\end{align}
 %Let us define $\rho \in \left[0,1\right]$ as the parameter that represent the rate of the AP-ZF symbols, being this rate $\rho\log_2(P)$ bits.
%%%%%%%%%%%%%%%%%%%%%%%%%%%%%%%%%%%%%%%%%%%%%%%%%%%%%
%%%%%%%%%%%%%%%%%%%%%%%%%%%%%%%%%%%%%%%%%%%%%%%%%%%%%
%\subsection{Parallel Configuration}\label{subse:parallel_conf}
%%%%%%%%%%%%%%%%%%%%%%%%%%%%%%%%%%%%%%%%%%%%%%%%%%%%%
%%%%%%%%%%%%%%%%%%%%%%%%%%%%%%%%%%%%%%%%%%%%%%%%%%%%%
In the \emph{Parallel Configuration}, represented in Fig.~\ref{fig:fig_sim_case}, it holds that
\begin{align}
\gamma_{i,i}&=1,\ \forall i \in \{1,2\} \label{eq:parallel_ch_1},\\
\gamma_{i,k}&=\gamma, \ \forall i,k \in \{1,2\} \mid k \neq i.\label{eq:parallel_ch_2}
\end{align}
Therefore, the CSIT \emph{quality exponents} are limited by
\begin{align}
\alpha^{(j)}_i&\leq \gamma,\ \forall i,j \in \{1,2\} \label{eq:alpha_bound},
\end{align}
and we assume that each RX has the same CSI quality, i.e., $\alpha^{(j)} = \alpha^{(j)}_i\ \forall i\in\{1,2\} $. Hence, Theorem~\ref{theo:distributed_gdof} then gives 
\begin{equation}
\begin{aligned}
		\GDoF^{DCSIT}(\{\alpha^{(j)}_{i,k}\}_{i,j,k\in\{1,2\}}) &= 2 - \gamma + \alpha^{(1)}.
\end{aligned}
\end{equation}

\begin{figure}[ht]
 \centering
				\includegraphics[trim = 31mm 74mm 34mm 59mm, clip,width=0.3\textwidth]{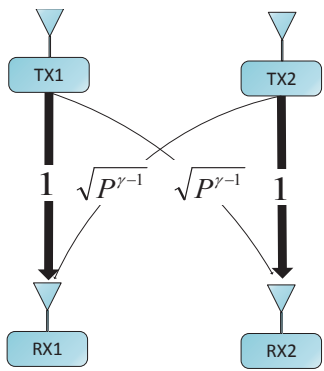}     
				\caption{Network topology for the \emph{Parallel Configuration}.} 
	\label{fig:fig_sim_case}			
\end{figure}

%%%%%%%%%%%%%%%%%%%%%%%%%%%%%%%%%%%%%%%%%%%%%%%%%%%%%
\subsection{Sketch of the proof}\label{subse:parallel_conf}
%%%%%%%%%%%%%%%%%%%%%%%%%%%%%%%%%%%%%%%%%%%%%%%%%%%%%
In the proposed transmission scheme, the transmitted symbols are
			\begin{align}
						\xv =  \frac{\sqrt{P-P^{1+\alpha^{(1)}-\gamma}}}{\sqrt{2}}\left[\begin{array}{c}
                                            1\\
                                            1
                                    \end{array}\right]s_0 + \bt^{\APZF}_1s_1 + \bt^{\APZF}_2s_2,
						\label{eq:transmited_sig}
			\end{align}
where
		\begin{itemize}
				\item $s_0 \in \mathbb{C}$ is a common symbol of rate $(\gamma-\alpha^{(1)})\log_2(P)$ bits that is decoded at both users. 
				\item $s_i \in \mathbb{C}$, with $i \in\{1,2\}$ is a  symbol of rate $(1+\alpha^{(1)}-\gamma)\log_2(P)$~bits intended to user~$i$. $\bt^{\APZF}_i\in\mathbb{C}^{2}$ is the AP-ZF precoder and this symbol is transmitted with power $\normv{\bt^{\APZF}_i}^2 \doteq P^{1+\alpha^{(1)}-\gamma}$.
		\end{itemize}    
			\begin{remark}
					As this work is focused on the high-SNR regime, the transmitted power of the common symbol always satisfies that $P-P^{1+\alpha^{(1)}-\gamma} \doteq P$. There exists still the need of the study of the power allocation in the finite-SNR regime.
			\end{remark}
%\begin{figure}[htb]
 %\centering
%\begin{subfigure}{.44\textwidth}
        %\centering
        %\includegraphics[trim = 31mm 74mm 34mm 59mm, clip,width=0.9\textwidth]{./scheme_2user_parallel.pdf}
         %\caption{}	\label{fig:fig_sim_case}			
    %\end{subfigure} 
    %\begin{subfigure}{0.44\textwidth}
        %\centering
        %\includegraphics[trim = 31mm 74mm 34mm 59mm, clip,width=0.9\textwidth]{./scheme_2user_tx_2.pdf}
         %\caption{}	\label{fig:fig_tx_case}
    %\end{subfigure}
     %\caption{Network topology for the \emph{Parallel Configuration} (a) and the \emph{Dominant-Uninformed-TX Configuration} (b).} 
%\end{figure}
%%\begin{figure}[ht]
 %%\centering
				%%\includegraphics[trim = 31mm 74mm 34mm 59mm, clip,width=0.3\textwidth]{./scheme_2user_parallel.pdf}
         %%\caption{}	\label{fig:fig_sim_case}			
     %%\caption{Network topology for the \emph{Parallel Configuration} (a).} 
%%\end{figure}
%%%%%%%%%%%%%%%%%%%%%%%%%%%%%%%%%%%%%%%%%%%%%%%%%%%%%
% \subsubsection*{Decoding}
%%%%%%%%%%%%%%%%%%%%%%%%%%%%%%%%%%%%%%%%%%%%%%%%%%%%%
The received signal at RX~$1$ is
		\begin{align}
				y_1&= \underbrace{\bh^\He_1 \frac{\sqrt{P-P^{1+\alpha^{(1)}-\gamma}}}{\sqrt{2}}\left[\begin{array}{c}
 1\\
 1\end{array}\right]s_0}_{\doteq \sqrt{P}} +  \underbrace{\bh^\He_1 \bt^{\APZF}_1 s_1 }_{\doteq \sqrt{P^{1+\alpha^{(1)}-\gamma}}}  +  \underbrace{\bh^\He_1 \bt^{\APZF}_2 s_2}_{\doteq \sqrt{P^0}}. \label{eq:achiev3c}
		\end{align}
The power scaling for $s_1$ comes from
		\begin{align}
        	\abs{\bh^\Herm_1{\bt^{\APZF}_1}}^2 &\doteq  P^{\tau-1} \max \left( P^{\gamma_{1,1}-(\gamma_{2,1}-\gamma_{2,2})^+}, P^{\gamma_{1,2}-(\gamma_{2,2}-\gamma_{2,1})^+}\right)\label{eq:received_power_rx1_1}\\
             &= P^{1+\alpha^{(1)}-\gamma-1}\max \left(P^{1-(\gamma-1)^+},P^{\gamma-(1-\gamma)^+}\right)\\
             &= P^{1+\alpha^{(1)}-\gamma},
		\end{align}
where  \eqref{eq:received_power_rx1_1} is obtained from applying %\eqref{eq:received_power} from
Corollary~\ref{cor:received}, with transmitted power $P^\tau = P^{1+\alpha^{(1)}-\gamma}$.
Also due to Corollary~\ref{cor:received}, the contribution of the interfering symbol $s_2$, lies on the noise floor thanks to the precoding: %\eqref{eq:received_interference} in Corollary~\ref{cor:received}, since it holds that
\begin{align}
						\abs{\bh^\He_1\bt^{\APZF}_2}^2 &\dotleq P^{\tau - 1}P^{(\min(\gamma_{1,1},\gamma_{1,2})-\alpha^{(1)})} \label{eq:interference_cancelation_parallel}\\
						&= P^{1+\alpha^{(1)}-\gamma - 1 + (\min(1,\gamma)-\alpha^{(1)})} \\
						&= P^0.
\end{align}
Fig.~\ref{fig:parallel_case_draw} illustrates the different power levels for the transmission towards RX~$1$. It shows that TX~$1$ reduces his transmitted power for $s_2$ to compensate that the channel from TX~$2$ is weaker, so that the interference power received at RX~$1$ from both TXs has the same scaling. Hence, the non-intended symbol scales in $P^{\alpha^{(1)}}$ and therefore thanks to the AP-ZF precoding it is possible to entirely cancel the interference (see Lemma~\ref{lem:power_scaling}). Due to the symmetry of the configuration, the received signal at RX~$2$ is studied in the same way.

We can see in \eqref{eq:achiev3c} that RX~$1$ receives the common symbol $s_{0}$ with a $\SNR$ scaling as $P^{\gamma-\alpha^{(1)}}$, treating $s_{1}$ as noise. After decoding the common symbol and removing its contribution to the received signal, $s_{1}$ has a $\SNR$ that scales as $P^{1+\alpha^{(1)}-\gamma}$. 
Likewise, \eqref{eq:achiev3c_general2} show that for RX~$2$ the common symbol $s_{0}$ has a $\SNR$ scaling as $P^{\gamma-\alpha^{(1)}}$, treating $s_{2}$ as noise. After decoding the common symbol and removing its contribution to the received signal, $s_{2}$ can be decoded, having it a $\SNR$ that scales as $P^{1+\alpha^{(1)}-\gamma}$.

Since symbols that are sent with a rate proportional to the SNR scaling can be decoded with a vanishing error probability, we can decode the common symbol $s_{0}$ with rate $(\gamma-\alpha^{(1)})\log_2(P)$ bits, $s_{1}$ and $s_2$ with rate $(1+\alpha^{(1)}-\gamma)\log_2(P)$ bits. That allows us to achieve a $\GDoF$ of
\begin{equation}
    \begin{aligned}
    		\GDoF^{DCSIT}(\{\alpha^{(j)}_{i,k}\}_{i,j,k\in\{1,2\}}) 	&= \gamma - \alpha^{(1)} + 2(1+\alpha^{(1)}-\gamma) \\
								&= 2+\alpha^{(1)}-\gamma.
    \label{eq:gdof_achievement}
    \end{aligned}
\end{equation}
This corresponds to the GDoF of the Centralized CSIT (See Theorem~\ref{theo:distributed_gdof}).

\begin{figure}[t]
 \centering
        \includegraphics[trim = 40mm 70mm 37mm 77mm, clip,width=0.6\textwidth]{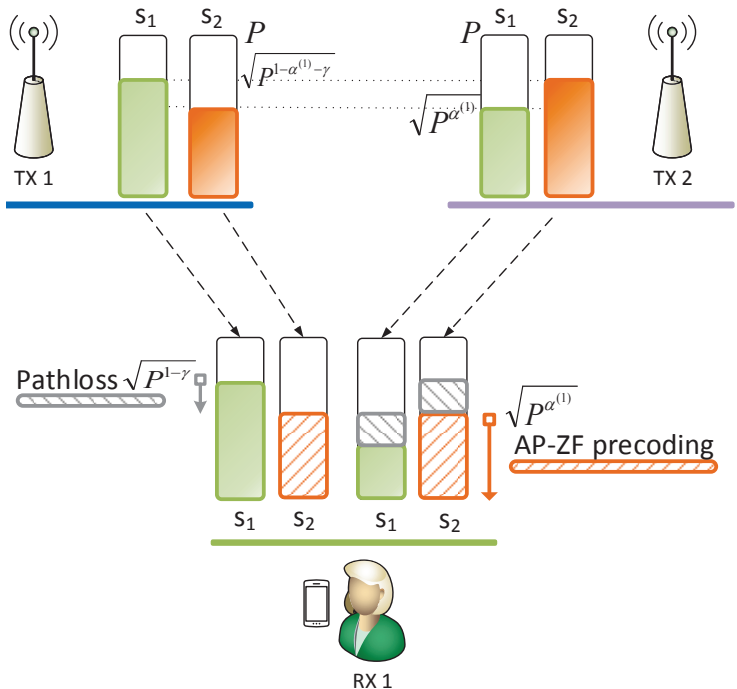}
        \caption{Illustration of the different power scaling for the \emph{Parallel Configuration} setting. Attenuation of the signal power due to the pathloss and the AP-ZF precoding are emphasized  using arrows.}
				\label{fig:parallel_case_draw}
\end{figure}
\FloatBarrier
\section{Simulation Results for the Parallel Configuration}\label{se:simulation} 
%%%%%%%%%%%%%%%%%%%%%%%%%%%%%%%%%%%%%%%%%%%%%%%%%%%%%
%%%%%%%%%%%%%%%%%%%%%%%%%%%%%%%%%%%%%%%%%%%%%%%%%%%%%
We now present some simulation results illustrating our main results. 
We consider the parallel topology introduced earlier in Section~\ref{se:achievability_simple}, with the coefficients:
\begin{align}
\gamma_{i,i} = 1,\ \ \  &\ \forall i\in \{1,2\},\\
\gamma_{i,k} = 0.8, &\ \forall i,k \in \{1,2\},k\neq i,
\end{align}
We further consider that TX~$1$ has the CSIT quality $\alpha^{(1)} = 0.5$ while TX~$2$ has $\alpha^{(2)} = 0$, i.e., no CSIT in terms of GDoF. 

The AP-ZF scheme has been simulated and compared with two different schemes. The first one is the Centralized CSIT setting where both TXs share the CSIT information, which has been shown in Theorem~\ref{theo:centralized_outerbound_proof} to be an outerbound. The second one is the naive distributed Zero-Forcing, where the discrepancies between TXs are not taken into account such that each TX implicitly assumes that the other TX has the same channel estimate \cite{dekerret2012_TIT}. 

\begin{figure}[th]
        \centering
        \includegraphics[trim = 40mm 85mm 40mm 90mm, clip,width=0.7\textwidth]{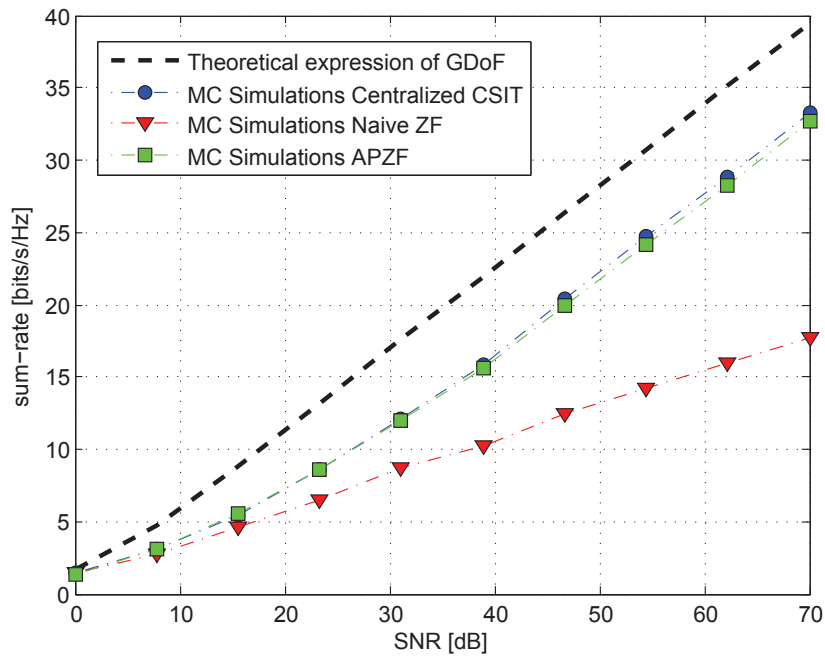}
        \caption{Sum rate in terms of the SNR for the \emph{Parallel Configuration} of Section~\ref{subse:parallel_conf}, with $\alpha^{(1)}=0.5$, $\alpha^{(2)}=0$ and $\gamma=0.8$.}
				\label{fig:plot_a_gamma}
\end{figure}
In Fig.~\ref{fig:plot_a_gamma}, the $\GDoF$ is equal to the slope at high SNR of the sum-rate function over the SNR. It can be seen that AP-ZF in the Distributed CSIT setting achieves the same GDoF of the Centralized CSIT case. Besides this, both cases achieve the theoretic $\GDoF$.
The gap between the outer bound and the simulations comes from the fact that the $\GDoF$ metric does not take into account the finite gaps, since they do not increase as function of $P$ (see \eqref{eq:sysmod_gdof}). 
The naive distributed Zero-Forcing is limited by the worst CSIT quality estimate, $\alpha^{(2)}=0$, and thus the CSIT at the best TX is useless for this naive ZF and it matches the performance of the setting with no CSIT \cite{dekerret2012_TIT}.

%%%%%%%%%%%%%%%%%%%%%%%%%%%%%%%%%%%%%%%%%%%%%%%%%%%%%
%%%%%%%%%%%%%%%%%%%%%%%%%%%%%%%%%%%%%%%%%%%%%%%%%%%%%
%%%%%%%%%%%%%%%%%%%%%%%%%%%%%%%%%%%%%%%%%%%%%%%%%%%%%
%%%%%%%%%%%%%%%%%%%%%%%%%%%%%%%%%%%%%%%%%%%%%%%%%%%%%
\section{Conclusion}\label{se:conclusions}
%%%%%%%%%%%%%%%%%%%%%%%%%%%%%%%%%%%%%%%%%%%%%%%%%%%%%
%%%%%%%%%%%%%%%%%%%%%%%%%%%%%%%%%%%%%%%%%%%%%%%%%%%%%
%%%%%%%%%%%%%%%%%%%%%%%%%%%%%%%%%%%%%%%%%%%%%%%%%%%%%
%%%%%%%%%%%%%%%%%%%%%%%%%%%%%%%%%%%%%%%%%%%%%%%%%%%%%
For the $2$-user MISO BC scenario with Distributed CSIT setting, with one TX being more informed than the other, we have shown that the discrepancies between CSI estimates at TXs do not reduce the $\GDoF$ performance, and that it is possible to achieve the $\GDoF$ of the genie-aided Centralized CSIT setting by a suitable adjustment of the transmitted power at each TX. We have developed an achievable scheme extended from the Active-Passive Zero-Forcing scheme presented in \cite{dekerret2016_ISIT_journal}, whose performance matches the centralized outer bound performance. Providing simulations results in a realistic environment, studying how to optimally reduce the finite gap that does not appears in the $\GDoF$ analysis, as well as the extension towards a setting with $K$ users, are open problems that constitute the next steps for the future research.

 \begin{appendices}
%%%%%%%%%%%%%%%%%%%%%%%%%%%%%%%%%%%%%%%%%%%%%%%%%%%%%
%%%%%%%%%%%%%%%%%%%%%%%%%%%%%%%%%%%%%%%%%%%%%%%%%%%%%
%%%%%%%%%%%%%%%%%%%%%%%%%%%%%%%%%%%%%%%%%%%%%%%%%%%%%
%%%%%%%%%%%%%%%%%%%%%%%%%%%%%%%%%%%%%%%%%%%%%%%%%%%%%
\section{Proof of Theorem~\ref{theo:distributed_gdof}: Achievability}\label{se:general_achievability}
%%%%%%%%%%%%%%%%%%%%%%%%%%%%%%%%%%%%%%%%%%%%%%%%%%%%%
%%%%%%%%%%%%%%%%%%%%%%%%%%%%%%%%%%%%%%%%%%%%%%%%%%%%%
%%%%%%%%%%%%%%%%%%%%%%%%%%%%%%%%%%%%%%%%%%%%%%%%%%%%%
%%%%%%%%%%%%%%%%%%%%%%%%%%%%%%%%%%%%%%%%%%%%%%%%%%%%%
In contrast to the main part of the article, and to preserve the symmetry between the TXs, %For convenience of the proof, and due to the fact that the expressions derived from Lemma~\ref{lem:power_scaling} do not depend on which TX is the active one, 
no assumption on which TX has the most accurate CSIT is done in this appendix. We then denote the best CSIT accuracy across TXs as 
		\begin{align}
				\{\alpha^{\max}_{i,k}\} \triangleq \begin{cases}
																			\{\alpha^{(1)}_{i,k}\} \mformtab \text{if $\{\alpha^{(1)}_{i,k}\} \geq \{\alpha^{(2)}_{i,k}\} \qquad \forall i,k\in\{1,2\}$},\\
																			\{\alpha^{(2)}_{i,k}\} \mformtab \text{if $\{\alpha^{(2)}_{i,k}\} > \{\alpha^{(1)}_{i,k}\} \qquad \forall i,k\in\{1,2\}$},
																	\end{cases}		
		\end{align}
and thus we can define the distributed counterparts of the terms $\alpha_1$, $\alpha_2$ from Theorem~\ref{theo:centralized_outerbound} as
		\begin{align}
				\alpha^{\prime}_1 \triangleq 	\min(\alpha^{\max}_{1,1},\alpha^{\max}_{1,2}), \\
				\alpha^{\prime}_2 \triangleq 	\min(\alpha^{\max}_{2,1},\alpha^{\max}_{2,2}).
		\end{align}		
We can assume w.l.o.g. that $\gamma_{1,1}$ is the strongest channel, i.e.,
		\begin{align}
				\gamma_{1,1} \geq \max(\gamma_{1,2},\gamma_{2,1},\gamma_{2,2}).		\label{eq:gamma_condition}	
		\end{align}
For ease of comprehension we split the demonstration in two different network topologies.

%%%%%%%%%%%%%%%%%%%%%%%%%%%%%%%%%%%%%%%%%%%%%%%%%%%%%%
%%%%%%%%%%%%%%%%%%%%%%%%%%%%%%%%%%%%%%%%%%%%%%%%%%%%%%
\subsection{if $\gamma_{2,1} \leq \gamma_{2,2}$}\label{se:achievability_0}
%%%%%%%%%%%%%%%%%%%%%%%%%%%%%%%%%%%%%%%%%%%%%%%%%%%%%%
%%%%%%%%%%%%%%%%%%%%%%%%%%%%%%%%%%%%%%%%%%%%%%%%%%%%%%
In that case, the sum GDoF expression of Theorem~\ref{theo:distributed_gdof} reads as 
        	\begin{align}
				\GDoF^{DCSIT}(\{\alpha^{(j)}_{i,k}\}_{i,j,k\in\{1,2\}})) = \min\left(\gamma_{1,1} + (\gamma_{2,2}-\gamma_{1,2}+ \alpha^{\prime}_1)^+,\ \gamma_{2,2} + \gamma_{1,1} -\gamma_{2,1} +  \alpha^{\prime}_2\right).\label{eq:gdof_1}  
			\end{align}
In this regime, the information sent with power below $P^{1-\gamma_{2,2}}$ lies on the noise floor at RX~$2$. Then, we can transmit information to RX~$1$ with power $P^{1-\gamma_{2,2}}$ and rate $(\gamma_{1,1}-\gamma_{2,2})\log_2(P)$ bits without generating interference at RX~$2$. Over this symbol, we use AP-ZF scheme to transmit information to both users.~\\

%%%%%%%%%%%%%%%%%%%%%%%%%%%%%%%%%%%%%%%%%%%%%%%%%%%%%%
\subsubsection{Transmitted signal}\label{subsubse:achievability_trans}
%%%%%%%%%%%%%%%%%%%%%%%%%%%%%%%%%%%%%%%%%%%%%%%%%%%%%%
Let us define $\rho \in \left[0,1\right]$ as the parameter that represent the rate of the AP-ZF symbols, i.e., the rate is $\rho\log_2(P)$ bits. Omitting the time indices, the transmitted symbols are
			\begin{align}
						\xv =  \bt^{\BC}s_0 + \bt^{\APZF}_1s_1 + \bt^{\APZF}_2s_2 + \bt^{z} z_1 
						\label{eq:achiev1}
			\end{align}
where
		\begin{itemize}
				\item $s_0 \in \mathbb{C}$ is a common symbol of rate $(\gamma_{2,2}-\rho)\log_2(P)$~bits that is decoded at both users. The precoder $\bt^{\BC}$ is the uniform multicast precoder
			\begin{align}
						\bt^{\BC} \triangleq \frac{\sqrt{P-P^{\rho+1-\gamma_{2,2}}-P^{1-\gamma_{2,2}}}}{\sqrt{2}}\left[\begin{array}{c}
                                            1\\
                                            1
                                    \end{array}\right].
						\label{eq:t_bc_def}
			\end{align} 
				\item $s_i \in \mathbb{C}$, $\forall i \in\{1,2\}$ is a AP-ZF symbol intended to RX~$i$ of rate $\rho\log_2(P)$~bits, where
			\begin{align}
						\rho \triangleq \min\left((\gamma_{2,2}-\gamma_{1,2}+ \alpha^{\prime}_1)^+,\ \gamma_{2,2}  -\gamma_{2,1} +  \alpha^{\prime}_2\right).
						\label{eq:rho_def}
			\end{align}  
			The term $\bt^{\APZF}_i$ is the AP-ZF precoder for RX~$i$. The transmitted power is $\normv{\bt^{\APZF}_i}^2\doteq P^{\rho+1-\gamma_{2,2}}$, where it holds that $P^{\rho+1-\gamma_{2,2}} \leq P$.                
 				\item $z_1 \in \mathbb{C}$ is a symbol of rate $(\gamma_{1,1}-\gamma_{2,2})\log_2(P)$~bits that carries a symbol desired by RX~$1$, and it does not generate interference at the other user. $\bt^z$ is the matched precoder with power transmission $\normv{\bt^z}^2\doteq P^{1-\gamma_{2,2}}$.                 
		\end{itemize}~\\
%%%%%%%%%%%%%%%%%%%%%%%%%%%%%%%%%%%%%%%%%%%%%%%%%%%%%%
\subsubsection{Received signal}\label{subsubsec:achievability_dec}
%%%%%%%%%%%%%%%%%%%%%%%%%%%%%%%%%%%%%%%%%%%%%%%%%%%%%%
The received signal at RX~$1$ is
		\begin{align}
				y_1&= \underbrace{\bh^\He_1 \bt^{\BC}s_0}_{\doteq \sqrt{P^{\gamma_{1,1}}}} +  \underbrace{\bh^\He_1 \bt^{\APZF}_1 s_1 }_{\doteq \sqrt{P^{\gamma_{1,1} -\gamma_{2,2} + \rho}}}  +  \underbrace{\bh^\He_1 \bt^{z} z_1}_{\doteq \sqrt{P^{\gamma_{1,1} -\gamma_{2,2}}}}  +  \underbrace{\bh^\He_1 \bt^{\APZF}_2 s_2}_{\doteq \sqrt{P^0}}, \label{eq:achiev3c_general}
		\end{align}
where the power scale for $s_1$  comes from Corollary~\ref{cor:received} with transmitted power $P^\tau = P^{\rho+1-\gamma_{2,2}}$, since it holds that
\begin{equation}
	\begin{aligned}
						\abs{\bh^\Herm_1{\bt^{\APZF}_1} }^2   &\doteq  P^{\tau-1} \max(P^{\gamma_{1,1} -(\gamma_{2,1}-\gamma_{2,2})^+}   , P^{\gamma_{1,2} -(\gamma_{2,2}-\gamma_{2,1})^+} )\\
										 &\overset{(a)}{=}  P^{\rho+1-\gamma_{2,2}-1} \max(P^{\gamma_{1,1}}, P^{\gamma_{1,2} -(\gamma_{2,2}-\gamma_{2,1})} )\\
										 &\overset{(b)}{=}  P^{\rho+1-\gamma_{2,2}-1} P^{\gamma_{1,1}} \\
										 &=  P^{\gamma_{1,1}-\gamma_{2,2}+\rho},\label{eq:received_rx_1_att}
		\end{aligned}
	\end{equation}
where $(a)$ comes from the fact that $(\gamma_{2,1}-\gamma_{2,2})^+ = 0$ and $(\gamma_{2,2}-\gamma_{2,1})^+ = \gamma_{2,2}-\gamma_{2,1}$, while $(b)$ comes from the assumption $\gamma_{1,1} \geq \max(\gamma_{1,2},\gamma_{2,1},\gamma_{2,2})$. Also due to Corollary~\ref{cor:received}, the contribution of the interfering symbol $s_2$, lies on the noise floor thanks to the precoding: 
%The contribution of the interfering symbol $s_2$ lies on the noise floor thanks to the precoding. The interference cancellation in \eqref{eq:achiev3c_general} is demonstrated through Corollary~\ref{cor:received} since it holds that
\begin{equation}
	\begin{aligned}				
				\abs{\bh^\He_1{\bt^{\APZF}_{2}}}^2 &\dotleq P^{\tau-1}P^{\min(\gamma_{1,1},\gamma_{1,2})-\alpha^{\prime}_1}\\
							&\overset{(a)}{=} P^{\rho+1-\gamma_{2,2}-1}P^{\gamma_{1,2}-\alpha^{\prime}_1}\\				
							&\overset{(b)}{\leq} P^{0},	\label{eq:interference_rx_1_att}
		\end{aligned}
	\end{equation}
where $(a)$ comes from the power level transmitted $P^\tau = P^{\rho+1-\gamma_{2,2}}$ and the fact that $\gamma_{1,1}\geq\gamma_{1,2}$, while $(b)$ comes from the definition of $\rho$ in \eqref{eq:rho_def} since it holds that  $\rho \leq \gamma_{2,2}-\gamma_{1,2}+\alpha^{\prime}_1$. 
The received signal at RX~$2$ is studied in the same way. Hence
        \begin{align}
              y_2&= \underbrace{\bh^\He_2 \bt^{\BC}s_0}_{\sqrt{\doteq P^{\gamma_{2,2}}}} +  \underbrace{\bh^\He_2 \bt^{\APZF}_2 s_2 }_{\doteq \sqrt{P^{\rho}}}    +  \underbrace{\bh^\He_2 \bt^{\APZF}_1 s_1}_{\doteq \sqrt{P^0}} +  \underbrace{\bh^\He_2 \bt^{z} z_1}_{\doteq \sqrt{P^{0}}}, \label{eq:achiev3c_general2}
		\end{align}
        what can be demonstrated following the same steps as in \eqref{eq:received_rx_1_att} and \eqref{eq:interference_rx_1_att} for $s_1$. Hence
\begin{equation}
	\begin{aligned}
						\abs{\bh^\Herm_2{\bt^{\APZF}_2} }^2   &\doteq  P^{\tau-1} \max(P^{\gamma_{2,1} -(\gamma_{1,1}-\gamma_{1,2})^+}   , P^{\gamma_{2,2} -(\gamma_{1,2}-\gamma_{1,1})^+} )\\
										 &=  P^{\rho+1-\gamma_{2,2}-1} \max(P^{\gamma_{2,1}-(\gamma_{1,1}-\gamma_{1,2})}, P^{\gamma_{2,2}})\\
										 &=  P^{\rho+1-\gamma_{2,2}-1} P^{\gamma_{2,2}} \\
										 &=  P^{\rho}. \label{eq:received_rx_2_att}
		\end{aligned}
	\end{equation}
Focusing on the interfering signal $s_1$, it holds that
\begin{equation}
	\begin{aligned}				
				\abs{\bh^\He_2{\bt^{\APZF}_{1}}}^2 &\dotleq P^{\tau-1}P^{\min(\gamma_{2,1},\gamma_{2,2})-\alpha^{\prime}_2}\\
							&= P^{\rho+1-\gamma_{2,2}-1}P^{\gamma_{2,1}-\alpha^{\prime}_2}\\				
							&\leq P^{0},		\label{eq:received_rx_2_int}
		\end{aligned}
	\end{equation}
where \eqref{eq:received_rx_2_int} comes  as \eqref{eq:interference_rx_1_att} from the definition of $\rho$ in \eqref{eq:rho_def} since it holds that  $\rho \leq \gamma_{2,2}-\gamma_{2,1}+\alpha^{\prime}_2$. ~\\

\FloatBarrier
%%%%%%%%%%%%%%%%%%%%%%%%%%%%%%%%%%%%%%%%%%%%%%%%%%%%%%
\subsubsection{Decoding and Achievable GDoF}\label{subsubse:achievability_gdof}
%%%%%%%%%%%%%%%%%%%%%%%%%%%%%%%%%%%%%%%%%%%%%%%%%%%%%%
We can see in \eqref{eq:achiev3c_general} that RX~$1$ receives the common symbol $s_{0}$ with a $\SNR$ scaling as $P^{\gamma_{2,2}-\rho}$, treating $s_{1}$ and $z_1$ as noise. After decoding the common symbol and removing its contribution to the received signal, $s_{1}$ can be decoded treating $z_1$ as noise, having it a $\SNR$ that scales as $P^{\rho}$. And finally, $z_1$ is decoded after removing the symbol $s_1$ from the received signal.
Likewise, \eqref{eq:achiev3c_general2} show that for RX~$2$ the common symbol $s_{0}$ has a $\SNR$ scaling as $P^{\gamma_{2,2}-\rho}$, treating $s_{2}$ as noise. After decoding the common symbol and removing its contribution to the received signal, $s_{2}$ can be decoded, having it a $\SNR$ that scales as $P^{\rho}$.

Since symbols that are sent with a rate proportional to the SNR scaling can be decoded with a vanishing error probability, we can decode the common symbol $s_{0}$ with rate $(\gamma_{2,2}-\rho)\log_2(P)$ bits, $s_{1}$ and $s_2$ with rate $\rho\log_2(P)$ bits and $z_1$ with rate $(\gamma_{1,1}-\gamma_{2,2})\log_2(P)$ bits. That allows us to achieve a $\GDoF$ of
\begin{equation}
    \begin{aligned}
    		\GDoF^{DCSIT}(\{\alpha^{(j)}_{i,k}\}_{i,j,k\in\{1,2\}}) 	&= (\gamma_{2,2}-\rho) + (\gamma_{1,1}-\gamma_{2,2}) + 2\rho \\
								&= \rho + \gamma_{1,1} \\
            		&=  \min(\gamma_{1,1} + (\gamma_{2,2}-\gamma_{1,2}+ \alpha^{\prime}_1)^+,\gamma_{2,2} + \gamma_{1,1} -\gamma_{2,1} +  \alpha^{\prime}_2).
    \label{eq:gdof_achievement}
    \end{aligned}
\end{equation}
%Hence, the outer bound is tight and we achieve the same $\GDoF$ than in the centralized setting.
This corresponds to the GDoF of the Centralized CSIT (See Theorem~\ref{theo:distributed_gdof}).

%%%%%%%%%%%%%%%%%%%%%%%%%%%%%%%%%%%%%%%%%%%%%%%%%%%%%%
%%%%%%%%%%%%%%%%%%%%%%%%%%%%%%%%%%%%%%%%%%%%%%%%%%%%%%
%%%%%%%%%%%%%%%%%%%%%%%%%%%%%%%%%%%%%%%%%%%%%%%%%%%%%%
\subsection{if $\gamma_{2,2} \leq \gamma_{2,1}$}\label{se:achievability_0_2}
%%%%%%%%%%%%%%%%%%%%%%%%%%%%%%%%%%%%%%%%%%%%%%%%%%%%%%
%%%%%%%%%%%%%%%%%%%%%%%%%%%%%%%%%%%%%%%%%%%%%%%%%%%%%%
%%%%%%%%%%%%%%%%%%%%%%%%%%%%%%%%%%%%%%%%%%%%%%%%%%%%%%
In the other case, the sum GDoF expression given in Theorem~\ref{theo:centralized_outerbound} is 
        	\begin{equation}
        	\begin{aligned}
				\GDoF^{DCSIT}(\{\alpha^{(j)}_{i,k}\}_{i,j,k\in\{1,2\}}) &= \min \bigg(\!\gamma_{1,1} + \max \Big( (\gamma_{2,2} - \gamma_{1,2} +  \alpha^{\prime}_1)^+\!,(\gamma_{2,1} - \gamma_{1,1} +  \alpha^{\prime}_1)^+\Big),\bigg. \\ 
                &\bigg. \ \ \ \ \qquad\ \  \gamma_{1,1} + (\gamma_{2,1} - \gamma_{1,1} + \gamma_{1,2} - \gamma_{2,2})^+ +  \alpha^{\prime}_2 \bigg).\label{eq:gdof_2}  
			\end{aligned}
			\end{equation}
Similarly to the previous case, the information sent with power below $P^{1-\gamma_{2,1}}$ lies on the noise floor for RX~$2$. Thus, we transmit to RX~$1$ a non-interfering symbol with power $P^{1-\gamma_{2,1}}$ and rate $(\gamma_{1,1}-\gamma_{2,1})\log_2(P)$ bit. Above it, AP-ZF scheme is used to transmit.~\\

%%%%%%%%%%%%%%%%%%%%%%%%%%%%%%%%%%%%%%%%%%%%%%%%%%%%%%
\subsubsection{Transmitted signal}\label{subsubse:achievability_trans_2}
%%%%%%%%%%%%%%%%%%%%%%%%%%%%%%%%%%%%%%%%%%%%%%%%%%%%%%
Let us keep the definition of $\rho \in \left[0,1\right]$ as the rate-parameter for the AP-ZF symbols (i.e., the rate is $\rho\log_2(P)$ bits). Omitting the time indices, the transmitted symbols are
			\begin{align}
						\xv =  \bt^{\BC}s_0 + \bt^{\APZF}_1s_1 + \bt^{\APZF}_2s_2 + \bt^{z} z_1 ,
						\label{eq:transmitted_case2}
			\end{align}
where
		\begin{itemize}
				\item $s_0$ is a common symbol of rate $(\gamma_{2,1}-\rho)\log_2(P)$~bits that is decoded at both users. The precoder $\bt^{\BC}$ is the uniform multicast precoder
			\begin{align}
						\bt^{\BC} \triangleq \frac{\sqrt{P-P^{\rho+1-\gamma_{2,1} + \min(\gamma_{1,1}-\gamma_{1,2},\ \gamma_{2,1}-\gamma_{2,2})}-P^{1-\gamma_{2,1}}}}{\sqrt{2}}\left[\begin{array}{c}
                                            1\\
                                            1
                                    \end{array}\right].
						\label{eq:t_bc_def2}
			\end{align} 
				\item $s_i$, $\forall i \in\{1,2\}$ is a AP-ZF symbol intended to the user~$i$ of rate $\rho\log_2(P)$~bits, where
			\begin{equation}
				\begin{aligned}
						\rho &\triangleq \min\LB\max((\gamma_{2,2} - \gamma_{1,2} +  \alpha^{\prime}_1)^+,\ (\gamma_{2,1} - \gamma_{1,1} +  \alpha^{\prime}_1)^+),\right. \\  
                        &\left.\mformtab\ \ \ (\gamma_{2,1} - \gamma_{1,1} + \gamma_{1,2} - \gamma_{2,2})^+ +  \alpha^{\prime}_2 \RB.
						\label{eq:rho_def_case2}
				\end{aligned}
      \end{equation}
			The term $\bt^{\APZF}_i$ is the AP-ZF precoder for RX~$i$. It is transmitted with power 
								\begin{align}
										\normv{\bt^{\APZF}_i}^2\doteq P^{\rho+1-\gamma_{2,1} + \min(\gamma_{1,1}-\gamma_{1,2},\ \gamma_{2,1}-\gamma_{2,2})},
								\end{align}
			where $\normv{\bt^{\APZF}_i}^2$ satisfies that $P^{\rho+1-\gamma_{2,1} + \min(\gamma_{1,1}-\gamma_{1,2},\ \gamma_{2,1}-\gamma_{2,2})} \leq P$.
 				\item $z_1$ is a symbol of rate $(\gamma_{1,1}-\gamma_{2,1})\log_2(P)$~bits that carries a symbol desired by RX~$1$ and it does not generate interference at the other user. $\bt^z$ is the matched precoder with transmitted power $\normv{\bt^{z}}^2 \doteq P^{1-\gamma_{2,1}}$.
		\end{itemize}~\\

%%%%%%%%%%%%%%%%%%%%%%%%%%%%%%%%%%%%%%%%%%%%%%%%%%%%%%
\subsubsection{Received signal}\label{subsubse:achievability_dec_2}
%%%%%%%%%%%%%%%%%%%%%%%%%%%%%%%%%%%%%%%%%%%%%%%%%%%%%%
 
% The received signal at RX~$1$ and RX~$2$ is shown respectively in Fig.~\ref{fig:received_signal_1} and Fig.~\ref{fig:received_signal_2}. 
The decoding is done in the same way. The received signal at RX~$1$ is
		\begin{align}
				y_1&= \underbrace{\bh^\He_1 \bt^{\BC}s_0}_{\doteq \sqrt{P^{\gamma_{1,1}}}} +  \underbrace{\bh^\He_1 \bt^{\APZF}_1 s_1 }_{\doteq \sqrt{P^{\gamma_{1,1} -\gamma_{2,1} + \rho}}}  +  \underbrace{\bh^\He_1 \bt^{z} z_1}_{\doteq \sqrt{^{\gamma_{1,1} -\gamma_{2,1}}}}  +  \underbrace{\bh^\He_1 \bt^{\APZF}_2 s_2}_{\doteq \sqrt{P^0}}, \label{eq:received_user1_case2}
		\end{align}
where the power scale for $s_1$  comes from applying Corollary~\ref{cor:received} with transmitted power $P^\tau = P^{\rho+1-\gamma_{2,1} + \min(\gamma_{1,1}-\gamma_{1,2},\ \gamma_{2,1}-\gamma_{2,2})}$, since it holds that
\begin{equation}
	\begin{aligned}
						\abs{\bh^\Herm_1{\bt^{\APZF}_1} }^2   &\doteq  P^{\tau-1} \max(P^{\gamma_{1,1} -(\gamma_{2,1}-\gamma_{2,2})^+} , P^{\gamma_{1,2} -(\gamma_{2,2}-\gamma_{2,1})^+})\\
										 &\overset{(a)}{=}  P^{\rho+1-\gamma_{2,1} + \min(\gamma_{1,1}-\gamma_{1,2},\gamma_{2,1}-\gamma_{2,2})-1} \max(P^{\gamma_{1,1}-(\gamma_{2,1}-\gamma_{2,2})}, P^{\gamma_{1,2}})\\
										 &{=}  P^{\rho-\gamma_{2,1} + \min(\gamma_{1,1}-\gamma_{1,2},\ \gamma_{2,1}-\gamma_{2,2}) + \max(\gamma_{1,1}-\gamma_{2,1}+\gamma_{2,2}, \gamma_{1,2})}\\
										 &=  P^{\gamma_{1,1}-\gamma_{2,1}+\rho}, \label{eq:received_rx_1_att_case2}
		\end{aligned}
	\end{equation}
where $(a)$ comes from the fact that $(\gamma_{2,1}-\gamma_{2,2})^+ = \gamma_{2,1}-\gamma_{2,2}$ and $(\gamma_{2,2}-\gamma_{2,1})^+ = 0$. Focusing on the interference cancellation in \eqref{eq:received_user1_case2}, it holds that
\begin{equation}
	\begin{aligned}				
				\abs{\bh^\He_1{\bt^{\APZF}_{2}}}^2 &\dotleq P^{\tau-1}P^{\min(\gamma_{1,1},\gamma_{1,2})-\alpha^{\prime}_1}\\
							&= P^{\rho+1-\gamma_{2,1} + \min(\gamma_{1,1}-\gamma_{1,2},\gamma_{2,1}-\gamma_{2,2})-1}P^{\gamma_{1,2}-\alpha^{\prime}_1}\\				
							&\leq P^{0},		\label{eq:interference_rx_1_case2}
		\end{aligned}
	\end{equation}
where \eqref{eq:interference_rx_1_case2} comes  from the definition of $\rho$ in \eqref{eq:rho_def_case2} since it holds that  $\rho \leq \gamma_{2,1} - \min(\gamma_{1,1}-\gamma_{1,2},\gamma_{2,1}-\gamma_{2,2})-\gamma_{1,2}+\alpha^{\prime}_1$. 
The received signal at RX~$2$ is studied in the same way. Hence
        \begin{align}
              y_2&= \underbrace{\bh^\He_2 \bt^{\BC}s_0}_{\sqrt{\doteq P^{\gamma_{2,1}}}} +  \underbrace{\bh^\He_2 \bt^{\APZF}_2 s_2 }_{\doteq \sqrt{P^{\rho}}}    +  \underbrace{\bh^\He_2 \bt^{\APZF}_1 s_1}_{\doteq \sqrt{P^0}} +  \underbrace{\bh^\He_2 \bt^{z} z_1}_{\doteq \sqrt{P^{0}}}, \label{eq:achiev3c_general2_case2}
		\end{align}
and, similarly to the previous case, it holds that
\begin{equation}
	\begin{aligned}
						\abs{\bh^\Herm_2{\bt^{\APZF}_2} }^2   &\doteq  P^{\tau-1} \max(P^{\gamma_{2,1} -(\gamma_{1,1}-\gamma_{1,2})^+}   , P^{\gamma_{2,2} -(\gamma_{1,2}-\gamma_{1,1})^+} )\\
										 &=  P^{\rho+1-\gamma_{2,1} + \min(\gamma_{1,1}-\gamma_{1,2},\gamma_{2,1}-\gamma_{2,2})-1} \max(P^{\gamma_{2,1}-(\gamma_{1,1}-\gamma_{1,2})}, P^{\gamma_{2,2}})\\
										 &\overset{(a)}{=}  P^{\rho-\gamma_{2,1} + \min(\gamma_{1,1}-\gamma_{1,2},\gamma_{2,1}-\gamma_{2,2}) +\max(\gamma_{2,1}-\gamma_{1,1}+\gamma_{1,2},\gamma_{2,2}})\\
										 &=  P^{\rho}, \label{eq:received_rx_2_att_2}
		\end{aligned}
	\end{equation}
	%where $(a)$ comes from the fact that in this section it holds that $g$
Focusing on the interfering signal $s_1$, it holds that
\begin{equation}
	\begin{aligned}				
				\abs{\bh^\He_2{\bt^{\APZF}_{1}}}^2 &\dotleq P^{\tau-1}P^{\min(\gamma_{2,1},\gamma_{2,2})-\alpha^{\prime}_2}\\
							&= P^{\rho+1-\gamma_{2,1} + \min(\gamma_{1,1}-\gamma_{1,2},\gamma_{2,1}-\gamma_{2,2})-1}P^{\gamma_{2,1}-\alpha^{\prime}_2}\\				
							&\leq P^{0},		\label{eq:received_rx_2_int_2}
		\end{aligned}
	\end{equation}
where \eqref{eq:received_rx_2_int_2} comes as \eqref{eq:interference_rx_1_case2} from the definition of $\rho$ in \eqref{eq:rho_def_case2} since it holds that  $\rho \leq \gamma_{2,1} - \min(\gamma_{1,1}-\gamma_{1,2},\gamma_{2,1}-\gamma_{2,2})-\gamma_{2,1}+\alpha^{\prime}_2$.~\\
%%%%%%%%%%%%%%%%%%%%%%%%%%%%%%%%%%%%%%%%%%%%%%%%%%%%%%
\subsubsection{Decoding and Achievable GDoF}\label{subsubse:achievability_gdof_2}
%%%%%%%%%%%%%%%%%%%%%%%%%%%%%%%%%%%%%%%%%%%%%%%%%%%%%%
Hence, from \eqref{eq:received_user1_case2}, the common symbol $s_{0}$ has a $\SNR$ scaling as $P^{\gamma_{2,1}-\rho}$, treating $s_{1}$ and $z_1$ as noise. After decoding the common symbol and removing its contribution to the received signal, $s_{1}$ can be decoded treating $z_1$ as noise, having it a $\SNR$ that scales as $P^{\rho}$. And finally, $z_1$ is decoded after removing the symbol $s_1$ from the received signal. In the same way, from \eqref{eq:achiev3c_general2_case2}, the common symbol $s_{0}$ has a $\SNR$ scaling as $P^{\gamma_{2,1}-\rho}$, treating $s_{2}$ as noise. After decoding the common symbol and removing its contribution to the received signal, $s_{2}$ can be decoded, having it a $\SNR$ that scales as $P^{\rho}$.~\\

Similarly to the previous case in Section~\ref{subsubse:achievability_gdof}, we can decode each symbol with a rate proportional to the SNR scaling and hence the common symbol $s_{0}$ can be decoded with rate $(\gamma_{2,1}-\rho)\log_2(P)$ bits, $s_{1}$ and $s_2$ with rate $\rho\log_2(P)$ bits and $z_1$ with rate $(\gamma_{1,1}-\gamma_{2,1})\log_2(P)$ bits. That allows us to achieve a $\GDoF$ of
\begin{equation}
    \begin{aligned}
    		\GDoF^{DCSIT}(\{\alpha^{(j)}_{i,k}\}_{i,j,k\in\{1,2\}}) 	&=(\gamma_{2,1}-\rho) + (\gamma_{1,1}-\gamma_{2,1})  + 2\rho \\
		            &= \gamma_{1,1} + \rho \\
								&=  \min \LB \gamma_{1,1} + \max\left( (\gamma_{2,2} - \gamma_{1,2} +  \alpha^{\prime}_1)^+,(\gamma_{2,1} - \gamma_{1,1} +  \alpha^{\prime}_1)^+\right),\right. \\ 
                   &\left. \mformtab\ \  \gamma_{1,1} + (\gamma_{2,1} - \gamma_{1,1} + \gamma_{1,2} - \gamma_{2,2})^+ +  \alpha^{\prime}_2 \RB.
    \label{eq:gdof_achievement_2}
    \end{aligned}
\end{equation}
%Hence, the outer bound is tight and we achieve the same $\GDoF$ than in the centralized setting.
This corresponds to the GDoF of the Centralized CSIT (See Theorem~\ref{theo:distributed_gdof}) which concludes the proof.
\end{appendices}

\bibliographystyle{IEEEtran}
\bibliography{./Literature}
\end{document}